\documentclass[12pt]{article}
\usepackage[utf8]{inputenc}
\usepackage[english]{babel}

\usepackage{amsmath, amsfonts, amssymb, amsthm}
\usepackage{graphicx}

\usepackage{times}
\usepackage{bm}
\usepackage[authoryear]{natbib}
\usepackage[linesnumbered,ruled,vlined]{algorithm2e}
\usepackage{xcolor}

\makeatletter
\renewcommand{\algocf@captiontext}[2]{#1\algocf@typo. \AlCapFnt{}#2} % text of caption
\renewcommand{\AlTitleFnt}[1]{#1\unskip}% default definition
\def\@algocf@capt@plain{top}
\renewcommand{\algocf@makecaption}[2]{%
  \addtolength{\hsize}{\algomargin}%
  \sbox\@tempboxa{\algocf@captiontext{#1}{#2}}%
  \ifdim\wd\@tempboxa >\hsize%     % if caption is longer than a line
    \hskip .5\algomargin%
    \parbox[t]{\hsize}{\algocf@captiontext{#1}{#2}}% then caption is not centered
  \else%
    \global\@minipagefalse%
    \hbox to\hsize{0x\@tempboxa}% else caption is centered
  \fi%
  \addtolength{\hsize}{-\algomargin}%
}
\makeatother

\def\T{{ \mathrm{\scriptscriptstyle T} }}

\newtheorem{theorem}{Theorem}

\newtheorem{proposition}{Proposition}

\newtheorem{remark}{\normalfont \scshape Remark}
\newtheorem{definition}{Definition}
\newtheorem{lemma}{Lemma}

\begin{document}

\title{Exact simulation of max-stable processes}

\author{Cl\'ement Dombry\thanks{Universit\'e de Franche--Comt\'e, Laboratoire de Math\'ematiques de 
       Besan\c{c}on, UMR CNRS 6623, 16 Route de Gray, 25030 Besan\c{c}on cedex, 
			 France. Email: clement.dombry@univ-fcomte.fr},\; Sebastian Engelke\thanks{\'Ecole Polytechnique F\'ed\'erale de Lausanne, 
			 EPFL-FSB-MATHAA-STAT, 
       Station 8, 1015 Lausanne, Switzerland.
       Facult\'e des Hautes Etudes Commerciales, Universit\'e de Lausanne, 
			 Extranef, UNIL-Dorigny, 1015 Lausanne, Switzerland.
			 Email: sebastian.engelke@epfl.ch}\; and\; Marco Oesting\thanks{University of Twente, 
			 Faculty of Geo-Information Science and Earth 
       Observation, PO Box 217, 7500 AE Enschede, The Netherlands.
       Email: m.oesting@utwente.nl}
       }
\date{}

\maketitle

\begin{abstract}
Max-stable processes play an important role as models for spatial extreme
events. Their complex structure as the pointwise maximum over an infinite 
number of random functions makes simulation highly nontrivial. Algorithms based
on finite approximations that are used in practice are often not exact and 
computationally inefficient. We will present two algorithms for exact 
simulation of a max-stable process at a finite number of locations. The first
algorithm generalizes the approach by \citet{DM-2014} for Brown--Resnick 
processes and it is based on simulation from the spectral measure. The second 
algorithm relies on the idea to simulate only the extremal functions, that is, 
those functions in the construction of a max-stable process that effectively 
contribute to the pointwise maximum. We study the complexity of both algorithms
and prove that the second procedure is always more efficient. Moreover, we 
provide closed expressions for their implementation that cover the most popular
models for max-stable processes and extreme value copulas. For simulation on 
dense grids, an adaptive design of the second algorithm is proposed.
\end{abstract}

\textbf{Keywords}: exact simulation; extremal function; extreme value distribution;
max-stable process; spectral measure.

\section{Introduction}

Max-stable processes have become widely used tools to model spatial extreme 
events. Occurring naturally in the context of extremes as limits of maxima of 
independent copies of stochastic processes, they have found many applications
in environmental sciences; see for instance \citet{coles93}, \citet{bdhz08}, 
\citet{bd2011}, \citet{dav2012b}.

Any sample continuous max-stable process $Z$ with unit Fr\'echet margins on 
some compact domain $\mathcal{X}\subset\mathbb{R}^d$ is characterized by a 
point process representation \citep{dehaan-1984}
\begin{equation}\label{eq:spectralrep}
Z(x)=\max_{i\geq 1} \zeta_i \psi_i(x),\quad x\in\mathcal{X},
\end{equation}
where $\{(\zeta_i,\psi_i),i\geq 1\}$ is a Poisson point process on 
$(0,\infty)\times \mathcal{C}$ with intensity measure 
$\zeta^{-2}\mathrm{d}\zeta\times\nu(\mathrm{d}\psi)$ for some locally finite
measure $\nu$ on the space $\mathcal{C}=\mathcal{C}(\mathcal{X},[0,\infty))$ of
continuous non-negative functions on $\mathcal{X}$ such that 
\begin{equation}\label{eq:normalisation}
\textstyle \int_{\mathcal{C}}\psi(x)\nu(\mathrm{d}\psi)=1,\quad x\in\mathcal{X}.
\end{equation}
Figure \ref{fig:intro} shows a realization of $Z$ as a mixture of
different random functions of the above point process.
Due to this complex structure of max-stable processes, in many cases, analytical 
expressions are only available for lower-dimensional distributions and
related characteristics need to be assessed by simulations.
Moreover, non-conditional
simulation is an important part of conditional simulation procedures that can 
be used to predict extreme events given some additional information
\citep[see][for example]{DEMR13, oestingschlather14}. Thus, there is a need for
fast and accurate simulation algorithms. 
\begin{figure} \label{fig:intro}
 \centering \includegraphics[trim = 0mm 15mm 3mm 10mm, clip,scale=.4, width=0.7\textwidth]{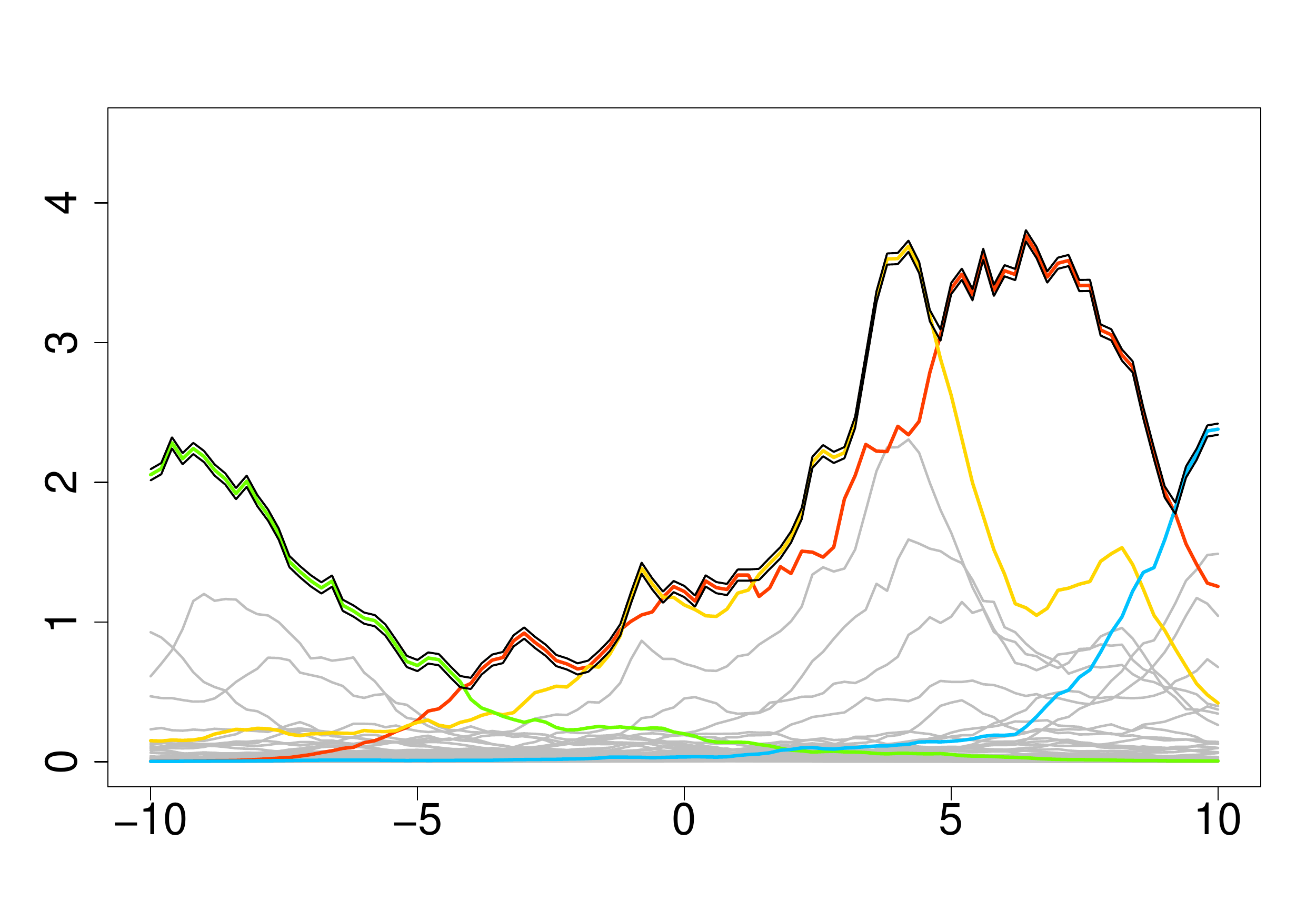}
 \caption{The Poisson point process $\{(\zeta_i,\psi_i),i\geq 1\}$ (grey). Only finitely many (colored) 
   $(\zeta_i,\psi_i)$ contribute the maximum process $Z$ (bordered in black).}
\end{figure}

As the spectral representation \eqref{eq:spectralrep} involves an infinite
number of functions, exact simulation of $Z$ is in general not straightforward
and finite approximations are used in practice. For the widely used
Brown--Resnick processes \citep{KSH-2009}, for instance, \citet{EKS-2011} and
\citet{OKS-2012} exploit the fact that the representation 
\eqref{eq:spectralrep} is not unique in order to propose simulation procedures 
based on equivalent representations. However, often these approximations do not 
provide satisfactory results in terms of accuracy or computational effort.

Exact simulation procedures can so far be implemented only in special cases. 
\citet{schlather-2002} proposes an algorithm that simulates the points 
$\{\zeta_i,i \geq 1\}$ in \eqref{eq:spectralrep} subsequently in a descending 
order until some stopping rule takes effect. If $\nu$ is the probability 
measure of a stochastic process whose supremum on $\mathcal{X}$ is almost 
surely bounded or if $Z$ is mixed moving maxima process with uniformly bounded
and compactly supported shape function, this procedure allows for exact 
simulation of $Z$. For extremal-$t$ processes \citep{opitz-2013}, the 
elliptical structure of Gaussian processes can be exploited to obtain exact 
samples \citep{thi2014}. \citet{OSZ-2013}  focus on a class of equivalent 
representations for general max-stable processes that, in principle, allow for
exact simulation in an optimal way in terms of efficiency. They propose to 
simulate max-stable processes via the normalized spectral representation with 
all the spectral functions sharing the same supremum. Being efficient with 
respect to the number of spectral functions, the simulation of a single 
normalized spectral function might be rather intricate in some cases including
the class of Brown--Resnick processes. For the latter, \citet{DM-2014} recently
proposed a representation that enables exact simulation at finitely many 
locations. 

Besides, several articles focus on the simulation of finite dimensional 
max-stable distributions or, equivalently, of their associated extreme value 
copula. \cite{gou1998} and \cite{cap2000} propose simulation procedures for 
certain bivariate extreme value distributions. \cite{ste2003} considers the 
case of extreme value distributions of logistic type. \cite{bol2009} provides a
method for exact simulation from the spectral measure of extremal Dirichlet and
logistic distributions.

In this paper, we propose two methods to simulate a general max-stable process
$Z$ exactly at a finite number of locations. At first, we propose a 
generalization of the algorithm by \citet{DM-2014} showing that their approach
relies on sampling from the spectral measure on the $L_1$-sphere of a 
multivariate extreme value distribution. The main idea of the second procedure
is to simulate out of the infinite set $\{\zeta_i\psi_i, i\geq 1\}$ only the 
extremal functions \citep[cf.][]{dombry-eyiminko-2012, 
dombry-eyiminko-2013}, i.e.\ those functions that satisfy 
$\zeta_i\psi_i(x)=Z(x)$ for some $x\in\mathcal{X}$ (the colored functions in 
Fig.~\ref{fig:intro}). In contrast to all existing simulation procedures, the
process $Z$ is not simulated simultaneously, but subsequently at different 
locations, rejecting all those functions that are not compatible with the 
process at the locations simulated so far. Both new procedures are based on
random functions following the same type of distribution that can be easily
simulated for most of the popular max-stable models. Our algorithms also apply 
very efficiently to exact simulation of finite-dimensional max-stable 
distributions or, equivalently, of the associated extreme value copulas.

\section{Extremal functions}

Without loss of generality, we may henceforth assume that $Z$ is a 
sample-continuous process with unit Fr\'echet margins given by the spectral 
representation \eqref{eq:spectralrep}. Indeed any sample-continuous max-stable 
process can be obtained from a process with unit Fr\'echet margins via marginal
transformations. We provide in this section some preliminaries on extremal 
functions and their distributions that will be essential in the new simulation 
methods. We use a point process approach and recall first that the 
$\mathcal{C}$-valued point process $\Phi=\{\phi_i\}_{i\geq 1}$ with 
$\phi_i=\zeta_i\psi_i$ is a Poisson point process with intensity
\begin{equation}\label{eq:defmu} \textstyle
\mu(A)=\int_{\mathcal{C}} \int_0^\infty 1_{\{\zeta \psi\in A\}}\zeta^{-2}\,\mathrm{d}\zeta\,\nu(\mathrm{d}\psi), \quad A\subset \mathcal{C} \mbox{ Borel}.
\end{equation}

\begin{definition}\label{def1}
Let $K\subset\mathcal{X}$ be a nonempty compact subset. A function $\phi\in\Phi$
is called $K$-extremal if there is some $x\in K$ such that $\phi(x)=Z(x)$, 
otherwise it is called $K$-subextremal. We denote by $\Phi_K^+$ the set of 
$K$-extremal functions and by $\Phi_K^-$ the set of $K$-subextremal functions. 
\end{definition}

It can be shown  that $\Phi_K^+$ and $\Phi_K^-$ are properly defined point 
process. When $K=\{x_0\}$, $x_0\in \mathcal{X}$, is reduced to a single point,
it is easy to show that $\Phi_{\{x_0\}}^+$ is also almost surely reduced to a
single point that we denote by $\phi_{x_0}^+$, termed the extremal function at 
point $x_0$. The distribution of $\phi_{x_0}^+$ is given in the next 
proposition \citep[see Proposition 4.2 in][]{dombry-eyiminko-2013}.

\begin{proposition}\label{prop2}
The random variables $Z({x_0})$ and $\phi_{x_0}^+/Z({x_0})$ are independent. 
Furthermore, $Z({x_0})$ has a unit Fr\'echet distribution and the distribution
of $\phi_{x_0}^+/Z({x_0})$ is
\begin{eqnarray}
P_{x_0}(A)={}&\mathrm{pr}(\phi_{x_0}^+/Z({x_0})\in A) 
={}& \int_{\mathcal{C}} 1_{\{f/f({x_0})\in A\}} f({x_0})\, \nu(\mathrm{d}f)\label{eq:prop2}, 
\quad A\subset \mathcal{C} \mbox{ Borel}.
\end{eqnarray}
\end{proposition}

By definition, $\phi_{x_0}^+({x_0})=Z({x_0})$. This entails that the 
distribution $P_{x_0}$ is supported by the subset of functions 
$\{f\in\mathcal{C}, f({x_0})=1\}$. 
\begin{proposition}\label{prop3}
The restricted point process 
\[
\Phi\cap\{f\in \mathcal C, f({x_0})>0\} 
\]
is a Poisson point process with intensity
\begin{align}
\int_{A} 1_{\{f({x_0})>0\}} \mu(\mathrm{d}f)
={}& \int_{\mathcal{C}} \int_0^\infty 1_{\{\zeta f\in A\}}\zeta^{-2}\,\mathrm{d}\zeta\, P_{x_0}(\mathrm{d}f),
 \quad A\subset \mathcal{C} \mbox{ Borel}\label{eq:prop3}. 
\end{align}
\end{proposition}

\begin{proof}
The fact that the restricted point process 
$\Phi\cap\{f\in \mathcal C, f({x_0})>0\}$ is a Poisson point process with 
intensity $1_{\{f({x_0})>0\}} \mu(\mathrm{d}f)$ is standard. We prove Equation
\eqref{eq:prop3}. For $A\subset\mathcal{C}$ Borel,
\begin{align*}
  \int_{\mathcal{C}} \int_0^\infty 1_{\{\zeta f\in A\}}\zeta^{-2}\mathrm{d}\zeta\,& P_{x_0}(\mathrm{d}f)
{}={} \int_{\mathcal{C}} \int_0^\infty 1_{\{\zeta f/f({x_0})\in A\}}\zeta^{-2}\mathrm{d}\zeta \,f({x_0})\,
       \nu(\mathrm{d}f) \displaybreak[0]\\
={}& \int_{\mathcal{C}} \int_0^\infty 1_{\{\tilde \zeta f\in A\}}\tilde\zeta^{-2}\mathrm{d}\tilde\zeta 
        \,1_{\{f({x_0})>0\}}\,\nu(\mathrm{d}f)
{}={} \int_{\mathcal{C}}1_{\{f\in A\}} 1_{\{f({x_0})>0\}}\mu(\mathrm{d}f).
\end{align*}
Here, we use successively Eq. \eqref{eq:prop2}, the change of variable 
$\tilde\zeta=\zeta/f({x_0})$ with $f({x_0})>0$ and Eq.\ \eqref{eq:defmu} for the
last equality.
\end{proof}

\begin{remark}\label{rkprop3}
As a consequence of \eqref{eq:prop3}, independent copies $Y_i$, $i\geq 1$,
of processes with distribution $P_{x_0}$ result in a point process 
$\{\zeta_iY_i\}_{i\geq 1}$ which has the same distribution as the
restricted point process $\Phi\cap\{f\in \mathcal C, f({x_0})>0\}$. If 
$\nu(\{f\in \mathcal C, f({x_0})=0\})=0$, then $\Phi$ consists only of 
functions with  positive value at ${x_0}$ and  $\Phi$ has the same distribution
as $\{\zeta_iY_i\}_{i\geq 1}$. This provides an alternative point process 
representation of the max-stable process $Z$ in terms of a random process $Y$ 
such that $Y({x_0})=1$ almost surely. In \cite{eng2012} and \cite{eng2014},
this representation is exploited for statistical inference of $Z$.
\end{remark}

It follows clearly from Definition \ref{def1} that we have the decomposition 
$\Phi=\Phi_K^+ \cup\Phi_K^-$. The following proposition will play a crucial 
role in our second approach based on extremal functions. If $f_1,f_2$ are two 
functions on $\mathcal{X}$, the notation $f_1<_Kf_2$ means $f_1(x)<f_2(x)$ for
all $x\in K$.
\begin{proposition}[\cite{dombry-eyiminko-2012}, Lemma 3.2]\label{prop1}
The conditional distribution of $\Phi_K^-$ with respect to $\Phi_K^+$ is equal
to the distribution of a Poisson point process on $\mathcal{C}$ with intensity
$1_{\{f<_K Z\}}\mu(df)$.
\end{proposition}

\section{Exact simulation procedures}

\subsection{Introduction}

Recently, \citet{DM-2014} provided a new representation of stationary
Brown--Resnick processes that allows for exact simulation of their 
finite-dimensional distributions. In this section we introduce two methods for
exact simulation of arbitrary max-stable processes and distributions.
More precisely, for a fixed number  $N\in\mathbb N$ of pairwise distinct 
locations $x=(x_1,\ldots,x_N)\in\mathcal{X}^N$, we aim at obtaining exact 
simulation of the max-stable random vector
\begin{align}\label{vec_Z}
  Z(x) = (Z(x_1),\dots, Z(x_N)).
\end{align}
Both procedures are intimately connected with the distribution $P_x$ in 
\eqref{eq:prop2}. 

The first method extends the \citet{DM-2014} approach. In fact, we show that 
their representation is nothing else than the so-called spectral 
representation of the Brown--Resnick process, and their procedure actually
enables exact simulation from the spectral measure on the $L_1$-sphere. 
In Section \ref{sec_spec} we derive a way of simulating from the spectral 
measure of a general max-stable distribution or a possibly non-stationary 
max-stable process.

The second procedure presented in Section \ref{sec_ppp} relies on conditional
distributions of the Poisson point process underlying any max-stable process. 
This approach also allows for exact simulation of \eqref{vec_Z} by simulating
at each location only the unique function that actually attains the maximum;
see Fig.~\ref{fig:intro}. It is intuitive and turns out to be even more 
powerful than the spectral method.

\subsection{Simulation via the spectral measure} \label{sec_spec}

Let us recall the spectral decomposition of the max-stable random vector 
$Z(x)$; for details we refer to \citet[Chap.\ 5]{res2008}. Here, we write
$f(x) = (f(x_1), \dots, f(x_N))$ for the restriction of a generic (random) 
function $f$ to the locations $x\in \mathcal X^N$. Following Equation 
\eqref{eq:spectralrep}, the max-stable random vector $Z(x)$ is generated by the
Poisson point process $\Phi_{x}=\{\zeta_i \psi_i(x), i\geq 1\}$ whose intensity
measure on the cone $E = [0,\infty)^N$ is denoted by $\mu_{x}$. Due to its 
homogeneity, the exponent measure $\mu_{x}$ can be factorized into a radial 
part on $(0,\infty)$ and an angular part on the unit $L_1$-sphere 
$S_{N-1}=\{z \in E: \, \|z\|=1\}$, where $\|z\| = z_1 + \dots + z_N$, for
$z = (z_1,\dots,z_N)\in E$. More precisely, a change to polar coordinates under
the map $U: E \to (0,\infty)\times S_{N-1}$, $U(z) = (\|z\|, z / \|z\|)$ yields
\begin{align} \label{spec_dens}
  \mu_{x}(F) &= \int_{U(F)} \mu_{x} \circ U^{-1}({\rm d}r, {\rm d} s) 
	        = N \int_{U(F)} r^{-2} {\rm d}r H({\rm d} s),
\end{align}
for any Borel subset $F\subset E$. The probability measure $H$ on $S_{N-1}$ is
called spectral measure of $Z(x)$ and it satisfies
\begin{align*}
  \textstyle \int_{S_{N-1}}  s_j H(d s) = N^{-1}, \quad j = 1,\dots, N.
\end{align*}
Equation \eqref{spec_dens} shows that we can represent the process $\Phi_{x}$ as
\begin{align*}
  \Phi_x = \{ U^{-1}(R_i, {Q}_i) :
i \geq 1\} = \{R_i {Q}_i : i \geq 1\},
\end{align*}
where $\{R_i:\ i \geq 1\}$ is a Poisson point process on $(0,\infty)$ with 
intensity $N r^{-2}{\rm d}r$ and ${Q}_i$, $i \geq 1$, are independent 
samples from the spectral measure $H$ on $S_{N-1}$. The great advantage of this 
representation is that the components of ${Q}_i$ are bounded by $1$. 
This ensures that 
$Z(x) = \max_{i  \geq 1} R_i {Q}_i$
can be simulated exactly by generating the largest $R_i$ first until no
more of the remaining points $R_i {Q}_i$ can contribute to the maximum.

The only difficulty is thus to generate the random variables ${Q}_i$
from the probability measure $H$ on the $(N-1)$-dimensional positive sphere 
$S_{N-1}$. The following theorem gives the solution to this problem for the
max-stable distribution $Z(x)$.

\begin{theorem} \label{theo:spec-measure}
  Let $T_i$, $i \geq 1$, be independent copies of a random variable $T$
  with uniform distribution on the discrete set $\{1,\dots, N\}$. Further, for
  any $k=1, \dots, N$, let ${Y}^{(k)}_i$, $i \geq 1$, be 
	independent random processes with distribution $P_{x_k}$ as in \eqref{eq:prop2}. 
Then, the $S_{N-1}$-valued random variables
  \begin{align*}
    {Q}_i =\frac{{Y}^{(T_i)}_i(x)}{\| {Y}^{(T_i)}_{i}(x)\|},\quad  i \geq 1,
  \end{align*}
  are independent with distribution $H$. Consequently, with 
  $\{R_i,\ i \geq 1\}$ as above, 
  \begin{align*}
    Z(x) = \max_{i \geq 1} R_i \frac{{Y}^{(T_i)}_i(x)}{\| {Y}^{(T_i)}_{i}(x)\|}.
  \end{align*}
\end{theorem}

\begin{proof}
 For any $k = 1,\dots, N$, Eq. \eqref{eq:prop2} implies 
  \begin{align}\label{eq:P_x}
   \textstyle \int_{\mathcal{C}} f(x_k) 1_{\{ f(x)  / \|f(x)\| \in A\}} \, \nu(\mathrm{d}f)
    = \int_{\mathcal{C}} 1_{\{ f(x)  / \|f(x)\| \in A\}} \, P_{x_k}(\mathrm{d}f).
  \end{align}  
  We compute the $\mu_{x}$-measure of the set $U^{-1}((u,\infty)\times A)$ for 
	$u>0$ and a Borel set $A \subset S_{N-1}$.
  \begin{align*}
       \mu_{x}(U^{-1}(&(u,\infty)\times A)) 
 {}={} \int_{\mathcal{C}} \int_0^\infty 1_{\{\zeta\|f(x)\|>u\}} 1_{\{f(x)/\|f(x)\|\in A\}} 
	           \zeta^{-2} \mathrm{d}\zeta\, \nu(\mathrm{d}f)\\
   ={}& \frac 1u \int_{\mathcal{C}} \|f(x)\| 1_{\{ f(x)  / \|f(x)\| \in A\}} \, \nu(\mathrm{d}f)
 {}={} \frac 1u \sum_{k=1}^N \int_{\mathcal{C}} f(x_k) 1_{\{ f(x)  / \|f(x)\| \in A\}} \, \nu(\mathrm{d}f)
        \displaybreak[0]\\
   ={}& \frac 1u \sum_{k=1}^N \int_{\mathcal{C}} 1_{\{ f(x)  / \|f(x)\| \in A\}} \, P_{x_k}(\mathrm{d}f)
 {}={} \frac Nu \cdot \frac 1N \sum_{k=1}^N \int_{\mathcal{C}} 1_{\{ f(x)  / \|f(x)\| \in A\}} \, 
	      P_{x_k}(\mathrm{d}f),
  \end{align*}  
  where the second last equation follows from \eqref{eq:P_x}. Let ${Y}^{(k)}$, 
	$k = 1, \dots, N$, be independent random processes with distribution 
	$P_{x_k}$, respectively, and let $T$ be an independent uniform random 
	variable on $\{1,\dots, N\}$, then the above implies that
  \begin{align*}
    \mu_{x}(U^{-1}((u,\infty)\times A)) 
    = \frac Nu \cdot \mathrm{pr}\left\{\frac{{Y}^{(T)}(x)}{\| {Y}^{(T)}(x)\|} \in A \right\}.
  \end{align*} 
  Comparing this with \eqref{spec_dens} yields the 
  assertion of the theorem.
\end{proof}

Theorem \ref{theo:spec-measure} shows how to simulate from the spectral measure
$H$. It requires only to be able to simulate from the distributions $P_{x_k}$, 
$k = 1,\dots,N$. Algorithm \ref{algo-spectral}, an adaptation of Schlather's
\citeyearpar{schlather-2002} algorithm, provides an exact sample from the 
max-stable process $Z$ at locations $x$.

\begin{algorithm} \caption{Simulation of a max-stable process $Z$, exactly at $x=(x_1,\ldots,x_N)$}
 \label{algo-spectral}
Simulate $\zeta^{-1}\sim \mathrm{Exp}(N)$ and set $Z(x)=0$.\\
While $(\zeta > \min(Z(x_1),\dots, Z(x_N)))\ \{$\\
\quad Simulate $T$ uniform on $\{1,\dots, N\}$ and $Y$ 
      according to the law $P_{x_T}$.\\
\quad Update $Z(x)$ by the componentwise
       $\max(Z(x),\zeta {Y}(x) / \|{Y}(x)\|)$.\\
\quad Simulate $E\sim \mathrm{Exp}(N)$ and update $\zeta^{-1}$ by $\zeta^{-1}+E$. \\
$\}$ \\
Return $Z$.
\end{algorithm}

\begin{remark}
 The results on the distribution $P_{x_k}$ for stationary Brown--Resnick
 processes obtained in Subsection \ref{subsec:BR} reveal that Algorithm
 \ref{algo-spectral} is identical to the algorithm by \citet{DM-2014} in this
 case.
\end{remark}

\subsection{Simulation via extremal functions}
\label{sec_ppp}

We now introduce the second procedure for exact simulation of the max-stable
process $Z$ at locations $x=(x_1,\ldots,x_N)\in\mathcal{X}^N$. For
$n=1,\ldots,N$ we consider the extremal and subextremal point processes 
$\Phi_n^+=\Phi_{\{x_1,\ldots,x_n\}}^+$ and $\Phi_n^-= \Phi_{\{x_1,\ldots,x_n\}}^-$.
We have $\Phi_n^+=\{\phi_{x_i}^+\}_{1\leq i\leq n}$ but there may be some 
repetitions in the right-hand side. We define the $n$th-step maximum process
\begin{equation}\label{eq:Zn}
Z_n(x)=\max_{\phi\in\Phi_n^+} \phi(x)=\max_{1\leq i\leq n} \phi_{x_i}^+(x),\quad x\in\mathcal{X}.
\end{equation}
By the definition of extremal functions we have $Z(x_i)=\phi_{x_i}^+(x_i)$ and 
clearly 
\begin{equation}\label{eq:Znbis}
Z(x_i)=Z_n(x_i), \quad i=1,\ldots,n.
\end{equation}
Hence, in order to exactly simulate $Z$ at locations $x$, it is enough
to exactly simulate $\Phi_N^+$. We will proceed inductively and simulate the 
sequence $(\phi_{x_n}^+)_{1\leq n\leq N}$ according to the following theorem. 
\begin{theorem}\label{theo2}
The distribution of $(\phi_{x_n}^+)_{1\leq n\leq N}$ is given as follows:
\begin{itemize}
\item {\it Initial distribution:}
the extremal function $\phi_{x_1}^+$ has the same distribution as $F_1Y_1$ 
where $F_1$ is a unit Fr\'echet random variable and $Y_1$ an independent
random process with distribution $P_{x_1}$ given by \eqref{eq:prop2}.
\item {\it Conditional distribution:}
for $1\leq n\leq N-1$, the conditional distribution of $\phi_{x_{n+1}}^+$ with
respect to $(\phi_{x_i}^+)_{1\leq i\leq n}$ is equal to the distribution of 
\[
\tilde\phi_{x_{n+1}}^+=\left\{\begin{array}{ll}
\mathop{\mathrm{argmax}}_{\phi\in\tilde\Phi_{n+1}} \phi(x_{n+1}) & \mbox{if } \tilde \Phi_{n+1}\neq\emptyset\\
\mathop{\mathrm{argmax}}_{\phi\in\Phi_n^+} \phi(x_{n+1}) & \mbox{if } \tilde \Phi_{n+1}=\emptyset
\end{array}\right.
\]
where $\tilde \Phi_{n+1}$ is a Poisson point process with intensity 
\begin{equation}\label{eq:tildephin+1}
1_{\left\{f(x_i)<Z_n(x_i),\ 1\leq i\leq n \right\}}1_{\left\{f(x_{n+1})>Z_n(x_{n+1})\right\}}\mu(\mathrm{d}f)
\end{equation}
and $Z_n$ is defined by \eqref{eq:Zn}.
\end{itemize}
\end{theorem}

\begin{proof}
The distribution of $\phi_{x_1}^+$ is given in Proposition \ref{prop2}. We 
prove the result for the conditional distribution of $\phi_{x_{n+1}}^+$ with
respect to $(\phi_{x_i}^+)_{1\leq i\leq n}$. Recall that 
$\Phi_n^+=\{\phi_{x_1}^+,\ldots,\phi_{x_n}^+\}$ and that according to 
Proposition \ref{prop1}, the conditional distribution of $\Phi_n^-$ with
respect to $\Phi_n^+$ is equal to the distribution of a Poisson point process
with intensity 
\begin{equation}\label{eq:intensity}
1_{\{f(x_i)< Z(x_i),\ 1\leq i\leq n\}}\mu(\mathrm{d}f)=1_{\{f(x_i)< Z_n(x_i),\ 1\leq i\leq n\}}\mu(\mathrm{d}f),
\end{equation}
where the equality follows from Eq.~\eqref{eq:Znbis}.
In order to determine $\phi_{x_{n+1}}^+$ we focus on the functions 
$\phi\in\Phi_n^-$ satisfying $\phi(x_{n+1})>Z_n(x_{n+1})$ and consider
the restriction 
\[
\tilde\Phi_{n+1}=\Phi_n^-\cap\left\{f\in\mathcal{C}, f(x_{n+1})>Z_n(x_{n+1})\right\}.
\]
It follows from Eq.~\eqref{eq:intensity} that conditionally on 
$(\phi_{x_i}^+)_{1\leq i\leq n}$, $\tilde\Phi_{n+1}$ is a Poisson point process
with intensity given by Eq.~\eqref{eq:tildephin+1}. We distinguish two cases:
\begin{itemize}
\item if $\tilde\Phi_{n+1}=\emptyset$ then there is no function in $\Phi_n^-$
exceeding $Z_n$ at point $x_{n+1}$, that is, $Z(x_{n+1})=Z_n(x_{n+1})$ 
and $\phi_{x_{n+1}}^+=\mathop{\mathrm{argmax}}_{\phi\in\Phi_n^+} \phi(x_{n+1})$.
\item If $\tilde\Phi_{n+1}\neq \emptyset$ then there is some function in 
$\Phi_n^-$ exceeding $Z_n$ at point $x_{n+1}$, that is, 
$Z(x_{n+1})>Z_n(x_{n+1})$ and 
$\phi_{x_{n+1}}^+=\mathop{\mathrm{argmax}}_{\phi\in\tilde\Phi_{n+1}} \phi(x_{n+1})$.
\end{itemize}
This concludes the proof of Theorem \ref{theo2}.
\end{proof}

\begin{figure}[t]
\centering
{\includegraphics[trim = 0mm 10mm 3mm 10mm, clip,scale=.35]{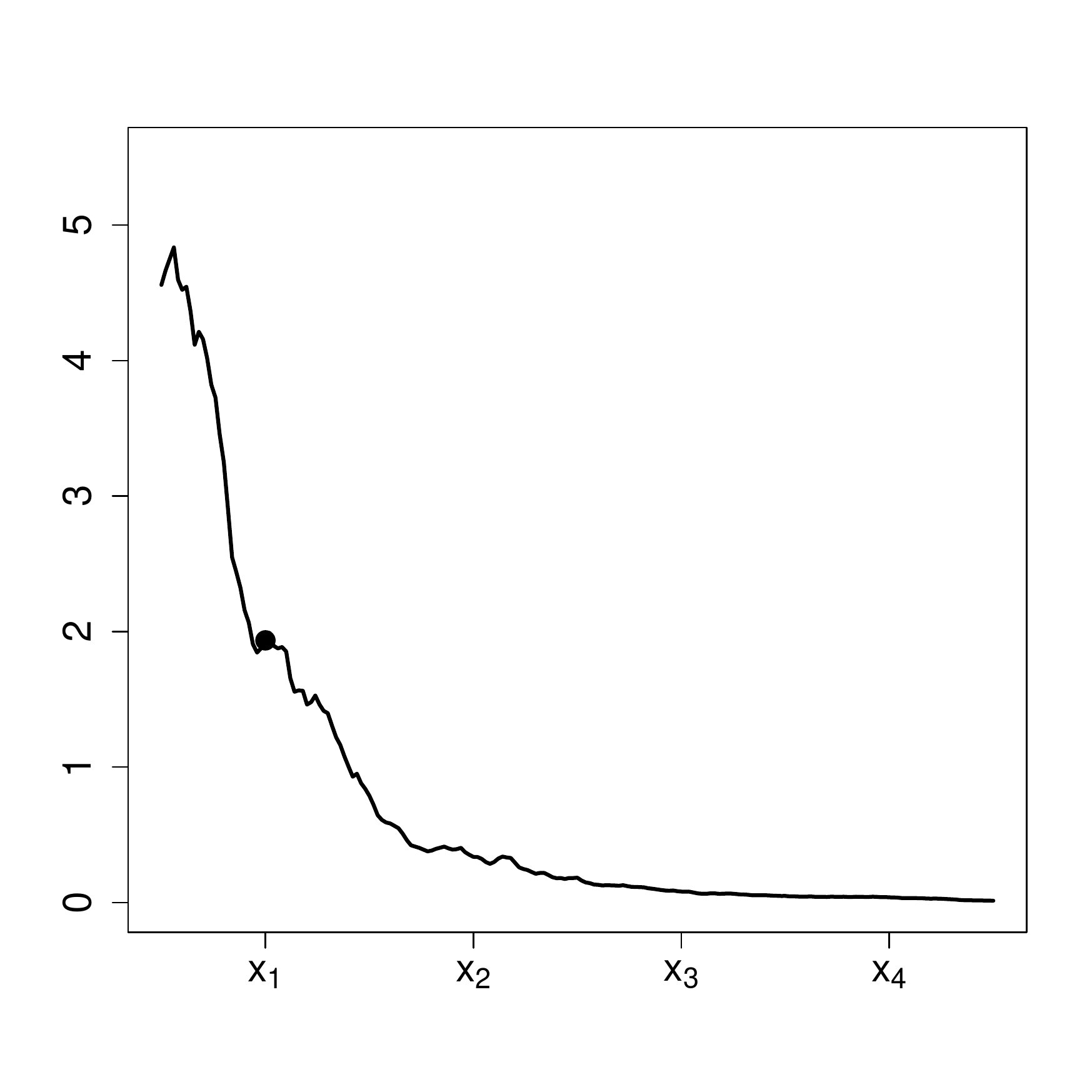}}%
{\includegraphics[trim = 10mm 10mm 3mm 10mm, clip,scale=.35]{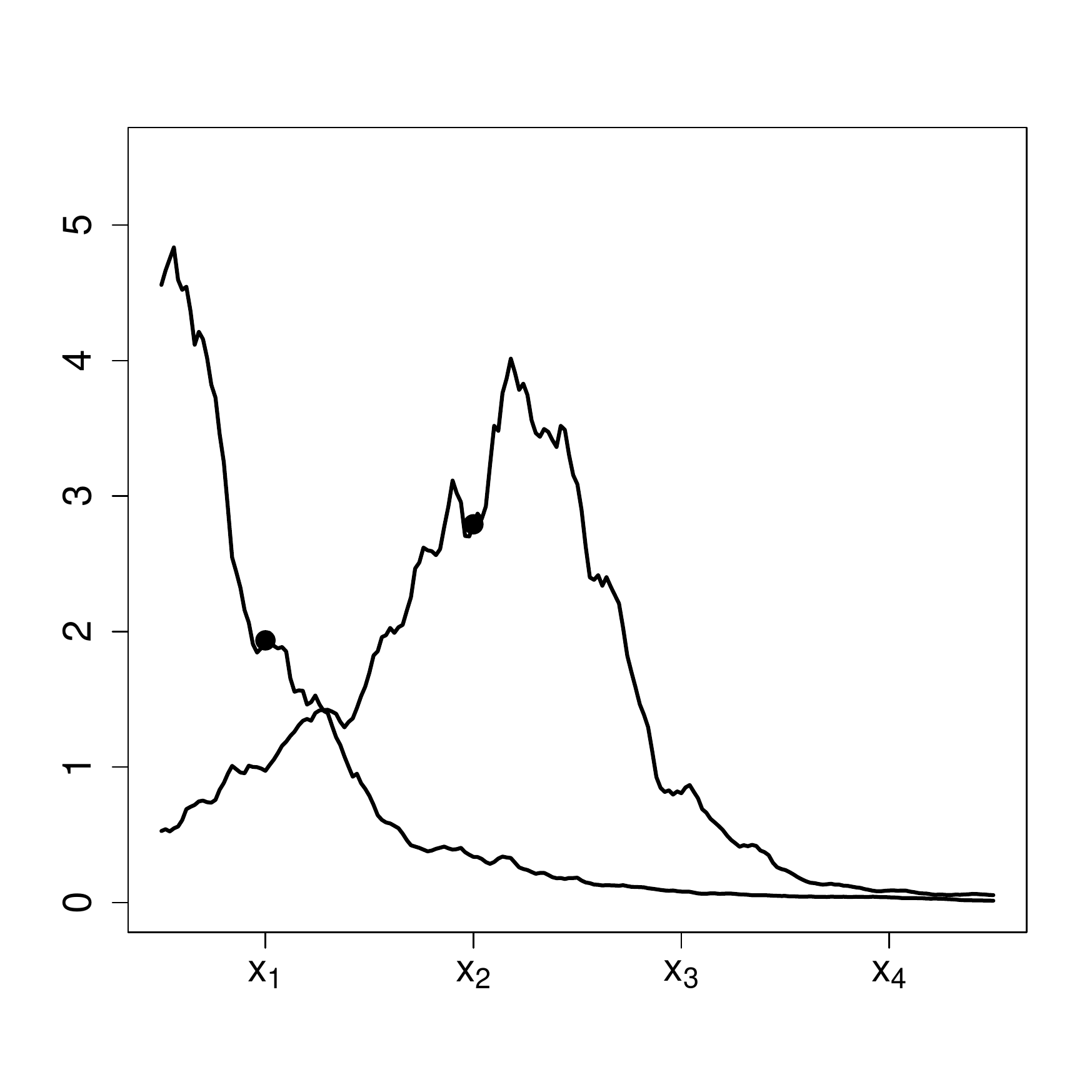}}\\
{\includegraphics[trim = 0mm 3mm 3mm 20mm, clip,scale=.35]{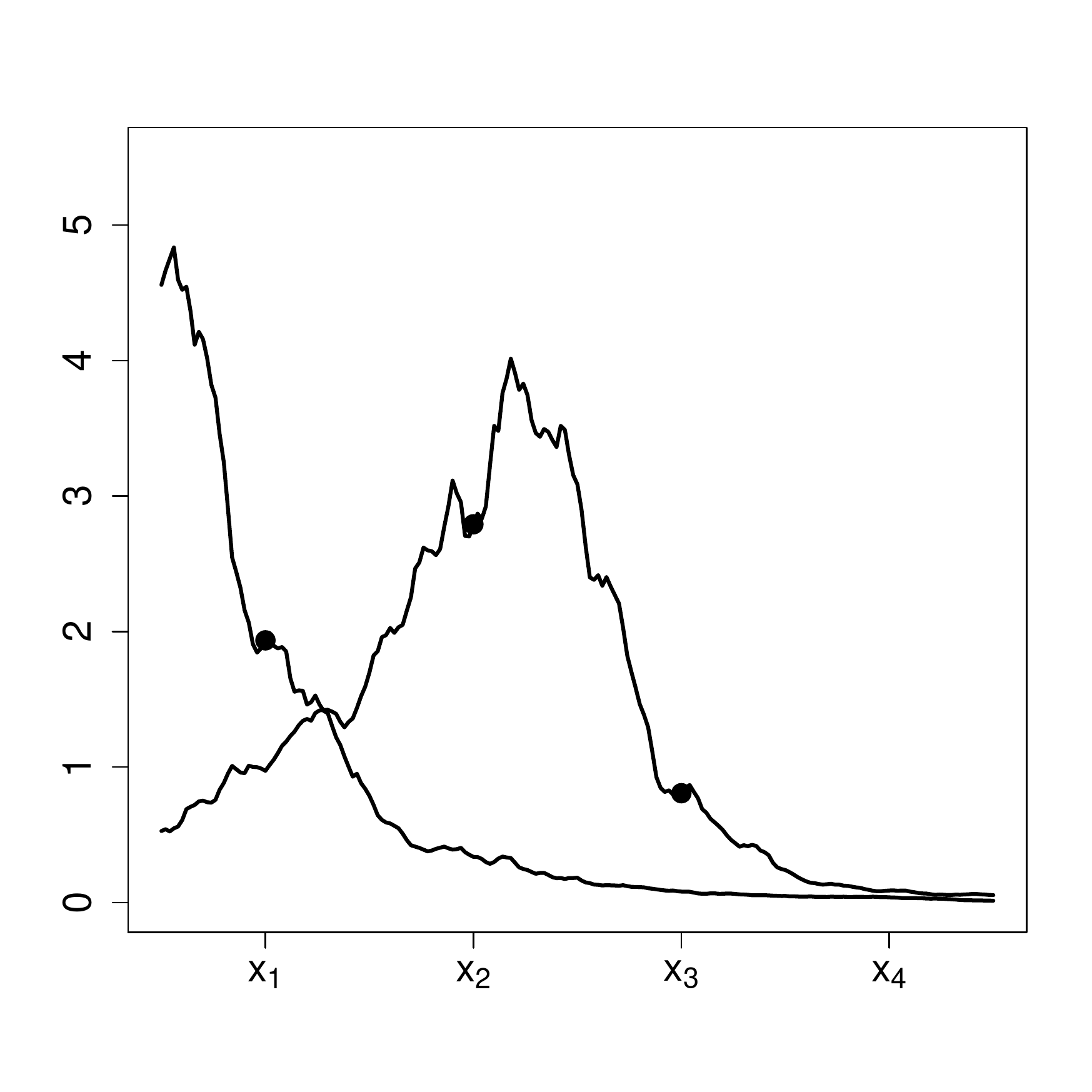}}%
{\includegraphics[trim = 10mm 3mm 3mm 20mm, clip,scale=.35]{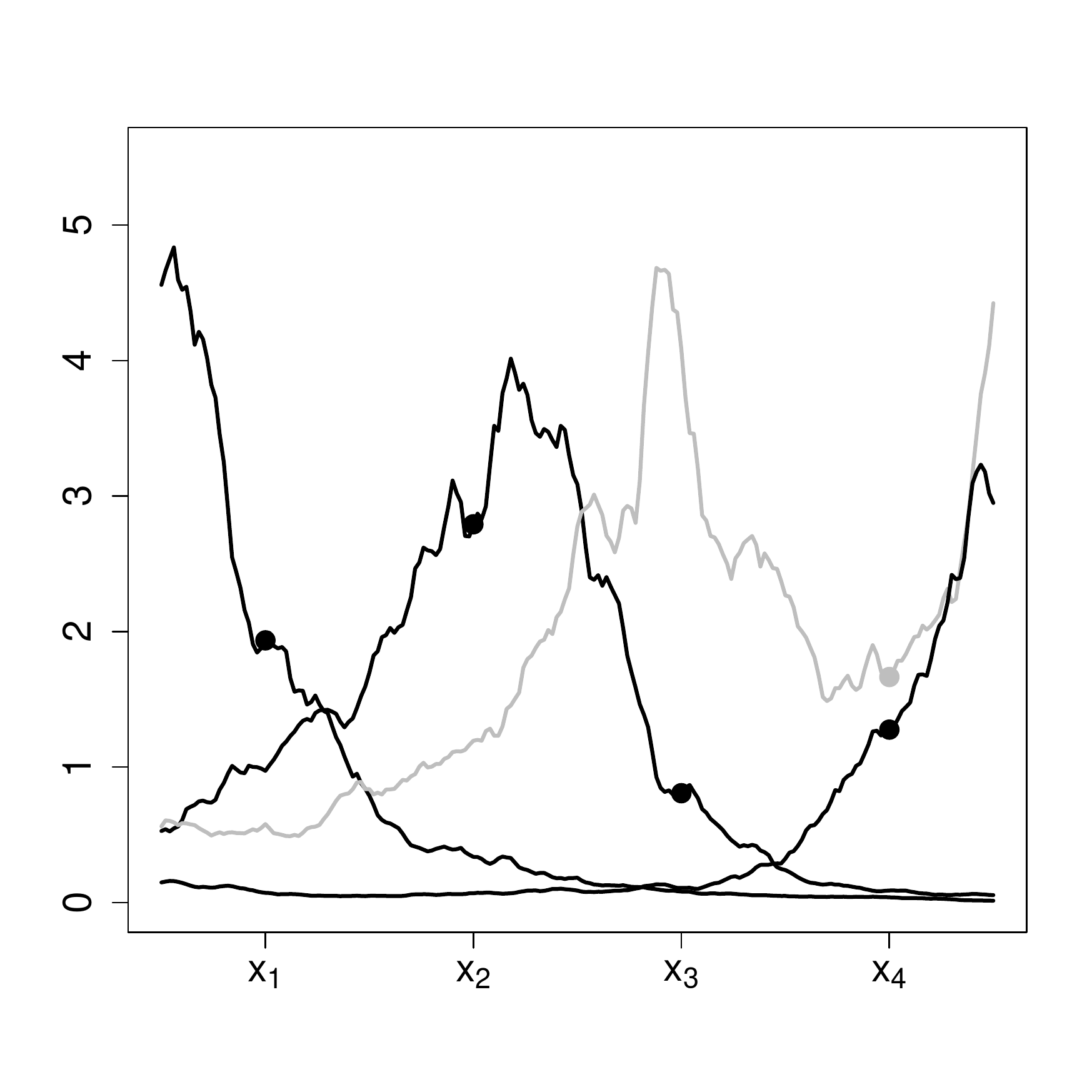}}%\subfloat[]
\caption{Simulation of $Z$ via Algorithm \ref{algo-ppp} at locations $(x_1,x_2,x_3,x_4)$.
  Initial process $\phi^+_{x_1}$ is always accepted (top-left).
  Second process $\phi^+_{x_2}$ is accepted as it exceeds $Z_1 = \phi^+_{x_1}$
  at $x_2$ but not at $x_1$ (top-right). Third process $\phi^+_{x_3}$ 
  is equal to $\phi^+_{x_2}$ since $\tilde \Phi_{3}=\emptyset$ (bottom-left).
  First sample of $P_{x_4}$ (grey line) is rejected since it exceeds $Z_3$ at
  $x_3$; second sample is valid and thus called $\phi^+_{x_4}$ (bottom-right). 
}
\label{fig:algo-ppp} 
\end{figure} 

From the above theorem, one can deduce Algorithm \ref{algo-ppp} for exact 
simulation of the max-stable process $Z$ at locations $x=(x_1,\ldots,x_N)$. 
According to Proposition \ref{prop3} and Remark \ref{rkprop3}, the
distribution $P_{x_{n+1}}$ can be used to simulate $\tilde\Phi_{n+1}$ with 
intensity \eqref{eq:tildephin+1}. Hence, as for the spectral method, the 
second algorithm requires only to be able to simulate from the distributions 
$P_{x_n}$, $n = 1,\dots, N$. Figure \ref{fig:algo-ppp} illustrates the 
procedure.

\begin{algorithm} \caption{Simulation of a max-stable process $Z$, exactly at $x = (x_1,\ldots,x_N)$} \label{algo-ppp}
Simulate $\zeta^{-1}\sim \mathrm{Exp}(1)$ and $Y\sim P_{x_1}$.\\
Set $Z(x)=\zeta Y(x)$.\\
For $n=2,\ldots,N$:\\
\quad  Simulate $\zeta^{-1}\sim \mathrm{Exp}(1)$.\\
\quad  while $(\zeta > Z(x_n))\ \{$\\
\quad\quad   Simulate $Y\sim P_{x_n}$.\\
\quad\quad   If $\zeta Y(x_i)< Z(x_i)$ for all  $i=1,\cdots,n-1$, \\
\quad\quad\quad      update $Z(x)$ by the componentwise $\max(Z(x),\zeta Y(x))$ .\\
\quad\quad   Simulate $E\sim \mathrm{Exp}(1)$ and update $\zeta^{-1}$ by $\zeta^{-1}+E$. \\
\quad   $\}$\\
Return $Z$.
\end{algorithm}

\section{Complexity of the Algorithms} \label{sec:complexity}

In this section, we aim at assessing the complexity of Algorithms
\ref{algo-spectral} and \ref{algo-ppp} as a function of the number $N$ of 
simulation sites. Both algorithms contain the simulation of exponential random
variables $E$ and the simulation of $N$-dimensional random vectors $Y(x)$ according to a
mixture of the laws $P_{x_1},\ldots,P_{x_N}$. The simulation of $E$ causes much
less computational effort than the simulation of $Y$ and can therefore be 
neglected in the analysis of the algorithmic complexity. We thus consider the  
number $C_1(N)$ and $C_2(N)$ of random vectors $Y(x)$ that need to be simulated
by Algorithm \ref{algo-spectral} and \ref{algo-ppp} respectively to obtain one 
exact simulation of $Z(x)$. Interestingly, one can provide simple expressions 
for $\mathrm{E}(C_1(N))$ and $\mathrm{E}(C_2(N))$.

\begin{proposition} \label{prop:complex}
 The expected number of random vectors $Y(x)$ that are needed for
 exact simulation of $Z$ at $x = (x_1,\ldots,x_N)$ are:\\[.5em]
  \emph{Algorithm 1:} $\quad \displaystyle{\mathrm{E}(C_1(N)) = N \mathrm{E}\left(\max_{i=1}^N Z(x_i)^{-1}\right)}$\\[.5em]
  \emph{Algorithm 2:} \quad $\mathrm{E}(C_2(N))=N$\\[.5em]
  Furthermore,  $\mathrm{E}(C_1(N)) \geq \mathrm{E}(C_2(N))$ with equality if 
	and only if $Z(x_1) = \ldots = Z(x_N)$ almost surely.
\end{proposition}

The expectation of $C_1(N)$ can be calculated similarly to Proposition 4.8 in
\citet{OSZ-2013}. The proof for the expectation of $C_2(N)$ is more difficult 
and, for the sake of brevity, it is provided as a supplementary material to 
this paper.

We conclude this section with some comments on the complexity of our algorithms
and a comparison with other exact simulation procedures. Proposition 
\ref{prop:complex} shows that, for any max-stable process, Algorithm 
\ref{algo-ppp} is more efficient than Algorithm \ref{algo-spectral} in terms of
the expected number of simulated functions. As the spectral functions follow 
either one of the laws $P_{x_1}, \ldots, P_{x_N}$ or a mixture of these, the 
simulation of a single spectral function is equally complex
in both cases. Thus, the new Algorithm \ref{algo-ppp} is always preferable to
the generalized Dieker--Mikosch algorithm, Algorithm \ref{algo-spectral}.

Next, we compare Algorithm \ref{algo-ppp} with the exact simulation algorithm 
via the normalized spectral representation proposed by \citet{OSZ-2013}. By 
Proposition 4.8 in \citet{OSZ-2013}, the number $C_3(N)$ of simulated 
normalized spectral functions satisfies
\begin{equation*}
 \mathrm{E}(C_{3}(N)) = \left(\int \max_{i=1}^N \psi(x_i) \nu({\rm d}\psi)\right)\mathrm{E}\left(\max_{i=1}^N Z(x_i)^{-1}\right)
\end{equation*}
and, thus, depends both on the geometry of the set $\{x_1,\ldots,x_N\}$ and on
the law of the max-stable process $Z$. For simulation on a large and dense 
subset of $\mathcal{X}$, the algorithm via the normalized spectral 
representation is more efficient than Algorithm \ref{algo-ppp} as 
$\mathrm{E}(C_{3}(N))$ is bounded by 
$\left(\int \sup_{x \in \mathcal{X}} \psi(x) \nu({\rm d}\psi)\right) 
\mathrm{E}\left(\sup_{x \in \mathcal{X}} Z(x)^{-1}\right)$
while $\mathrm{E} C_2(N)=N$ grows with the size of the subset. If $N$ is 
small or moderate and the vector $Z(x)$ is weakly dependent, we may also
have $\mathrm{E}(C_2(N))< \mathrm{E}(C_{3}(N))$.

Besides the efficiency in terms of the expected number of simulated functions,
we also need to take into account the complexity of the simulation of a single
spectral function. Exact and efficient simulation procedures for the 
normalized spectral function are known for some cases only (such as moving 
maxima processes), while they are not available in other cases like 
Brown--Resnick or extremal-$t$ processes. In contrast, the random functions in 
Algorithms \ref{algo-spectral} and \ref{algo-ppp} with
distributions $P_{x_0}$ in \eqref{eq:prop2}, $x_0\in\mathcal X$, 
can be simulated efficiently for the most popular max-stable process and extreme value copula models. Indeed, in Section \ref{sec:examples} below we provide closed
form expressions for various important examples.

\section{Examples} \label{sec:examples}
 
\subsection{Moving maximum process}

The parameter space is  $\mathcal{X}=\mathbb{Z}^d$ or $\mathbb{R}^d$ and 
$\lambda$ denotes the counting measure or the Lebesgue measure, respectively.
A moving maximum process on $\mathcal{X}$ is a max-stable process of the form
\begin{equation}\label{eq:MMP}
Z(x)=\max_{i\geq 1} \zeta_i h(x-\chi_i),\quad x\in\mathcal{X},
\end{equation}
where $\{(\zeta_i,\chi_i), i\geq 1\}$ is a Poisson point process on
$(0,\infty)\times\mathcal{X}$ with intensity measure 
$\zeta^{-2}\mathrm{d}\zeta\times \lambda(\mathrm{d}\chi)$ and 
$h:\mathcal{X}\to [0,\infty)$ is a continuous function satisfying 
$\int_{\mathcal{X}} h(x)\lambda(\mathrm{d}x)=1$. A famous example is the 
Gaussian extreme value process where $h$ is a multivariate Gaussian density on
$\mathbb{R}^d$ \citep{smith-1990}. 
\begin{proposition}\label{prop:mmp}
Consider the moving maximum process $\eqref{eq:MMP}$. For all $x_0\in\mathcal{X}$,
the distribution $P_{x_0}$ is equal to the distribution of the random function
\[
\frac{h(\cdot+\chi-x_0)}{h(\chi)}\quad \mbox{with} \ \chi\sim h(u)\lambda(\mathrm{d}u).
\]
\end{proposition}
All proofs of this section can be found in the supplementary material to this paper. 

\subsection{Brown--Resnick process} \label{subsec:BR}

We consider max-stable processes obtained by representation 
\eqref{eq:spectralrep} where $\nu$ is a probability measure on $\mathcal{C}$
given by
\begin{equation}\label{eq:BR}
\nu(A)=\mathrm{pr}\left[\exp\left(W(\cdot)-\sigma^2(\cdot)/2\right)\in A\right],\quad A\subset \mathcal{C}\ \mbox{Borel}
\end{equation}
with $W$ a sample-continuous centered Gaussian process on $\mathcal{X}$ with 
variance $\sigma^2(x)=\mathrm{E}[W(x)^2]$. In other words, $\nu$ is the 
distribution of the log-normal process 
$Y(x)=\exp\left(W(x)-\sigma^2(x)/2\right)$, $x\in \mathcal X$. The relation 
$\mathrm{E}[\exp\{W(x)\}]=\exp(\sigma^2(x)/2)$ ensures that $\mathrm{E}[Y(x)]=1$
and Equation \eqref{eq:normalisation} is satisfied.

An interesting phenomenon arises when $\mathcal{X}=\mathbb{Z}^d$ or 
$\mathbb{R}^d$ and $W$ has stationary increments: \citet{KSH-2009} show that 
the associated max-stable process $Z$ is then stationary with distribution 
depending only on the semi-variogram
\[
\gamma(h)=\frac{1}{2}\mathrm{E}\left[\{W(h)-W(0)\}^2\right],\quad h\in\mathcal{X}.
\]
The stationary max-stable process $Z$ is called a Brown--Resnick process. 
However, our results apply both in the stationary and non-stationary case 
\citep[cf.,][]{kab2011} and 
except stated otherwise we do not assume that $W$ has stationary increments.

\begin{proposition}\label{prop:BR}
Consider the Brown--Resnick type model $\eqref{eq:BR}$. For all 
$x_0\in\mathcal{X}$, the distribution $P_{x_0}$ is equal to the distribution of
the log-normal process
\[
\widetilde Y(x)=\exp\left(W(x)-W(x_0)-\frac{1}{2}\mathrm{Var}[W(x)-W(x_0)] \right),\quad x\in\mathcal{X}.
\]
\end{proposition}

\begin{remark}{\rm
It is easy to deduce from the proposition that in the Brown--Resnick case where
$W$ has stationary increments, then $\widetilde Y$ has the same distribution as 
\[
\exp\left(W(x-x_0)-W(0)-\frac{1}{2}\gamma(x-x_0)\right),\quad x\in\mathcal{X}.
\]
}
\end{remark}

\begin{remark}{\rm
    The finite dimensional margins of Brown--Resnick processes
    are H\"usler--Reiss distributions \citep[cf.,][]{huesler-reiss-1989}
    and the above therefore
    provides a method for their exact simulation.
}
\end{remark}

\subsection{Extremal-$t$ process}

We consider the so called extremal-$t$ max-stable process \citep[cf.,][]{opitz-2013} defined by representation
\eqref{eq:spectralrep} with $\nu$  the distribution of the random process
\begin{equation}\label{eq:extremal-t}
Y(x)=c_\alpha \max(0,W(x))^\alpha,\quad x\in\mathcal{X},
\end{equation}
where $\alpha > 0$, $c_\alpha={\pi}^{1/2}2^{-(\alpha-2)/2}/\Gamma\left(\frac{1+\alpha}{2}\right)$,
and $W$ a sample-continuous centered Gaussian process on $\mathcal{X}$ with 
unit variance and covariance function $c$. The constant $c_\alpha$ is such that $\mathrm{E}[Y(x)]\equiv 1$ so
that Equation \eqref{eq:normalisation} is satisfied. For $\alpha = 1$, the 
corresponding max-stable process in \eqref{eq:spectralrep} coincides with the
widely used extremal Gaussian process by \citet{schlather-2002}. 

\begin{proposition}\label{prop:extremal-t}
Consider the extremal-$t$ model $\eqref{eq:extremal-t}$. For all 
$x_0\in\mathcal{X}$, the distribution $P_{x_0}$ is equal to the 
distribution of $\max(T,0)^\alpha$, where $T=(T(x))_{x\in\mathcal{X}}$ is a
Student process with $\alpha+1$ degrees of freedom, location and scale 
functions given respectively by
\[
\mu(x)=c(x_0,x)\quad \mbox{and}\quad \hat c(x_1,x_2)=\frac{c(x_1,x_2)-c(x_0,x_1)c(x_0,x_2)}{(\alpha+1)}.
\]
\end{proposition}

\subsection{Multivariate extreme value distributions}

In this section, we review some popular models for multivariate extreme value
distributions, i.e., the case when $\mathcal X = \{1,\dots,N\}$ in \eqref{eq:spectralrep} is a finite set for some fixed $N\in\mathbb N$. 
For these \mbox{models}, we explicitly calculate the measure $P_{j_0}$ 
for any $j_0 = 1,\dots,N$.
Unless otherwise stated, all random vectors are $N$-dimensional in this section.
Multivariate extreme value distributions differ from extreme value copulas
only by a change in the marginal distribution, so that our methodology applies
directly to exact simulation of extreme value copulas. For more details on
the models, we refer to \citet{gudendorf-segers-2010}.
\medskip

\textbf{Logistic model}

The symmetric logistic model in dimension $N$ with parameter $\theta\in (0,1)$
corresponds to the max-stable random vector with cumulative distribution function
\begin{equation}\label{eq:logistic}
\mathrm{pr}[Z\leq z]=\exp\left(-\left(\sum\nolimits_{j=1}^N z_j^{-1/\theta} \right)^\theta\right),
        \quad z=(z_1,\ldots,z_N)\in(0,\infty)^N.
\end{equation}

\begin{proposition}\label{prop:logistic}
Let $\beta=1/\theta$. In the logistic model \eqref{eq:logistic}, the probability
measure $P_{j_0}$ for any $j_0 = 1,\dots,N$ is equal to the distribution of the random vector
\[
 \left(\frac{F_1}{F_{j_0}},\ldots,\frac{F_N}{F_{j_0}} \right)
\]
where $F_1,\ldots,F_N$ are independent, $F_j$, $j\neq j_0$, follows a 
$\mathrm{Frechet}(\beta,c_\beta)$ distribution with scale parameter
$c_\beta=\Gamma(1-1/\beta)^{-1}$ and $(F_{j_0}/c_\beta)^{-\beta}$ follows a 
$\mathrm{Gamma}(1-1/\beta,1)$ distribution.
\end{proposition}
\begin{remark}
  The asymmetric logistic distribution can be represented as the mixture of
  symmetric logistic distributions; see Theorem 1 in \cite{ste2003}, for instance.
  As a consequence, Proposition \ref{prop:logistic} also enables exact
  simulation of asymmetric logistic distributions.
\end{remark}

\textbf{Negative logistic model}

The negative logistic model in dimension $N$ with parameter $\theta>0$ 
corresponds to the max-stable random vector $Z$ with cumulative distribution
function
\begin{equation}\label{eq:neglogistic}
\mathrm{pr}[Z\leq z]=\exp\left( \sum_{\emptyset\neq J\subset\{1,\ldots,N\}}(-1)^{|J|} 
        \left(\sum\nolimits_{j\in J} z_j^{\theta} \right)^{-1/\theta}\right),\quad z\in(0,\infty)^N.
\end{equation}

\begin{proposition}\label{prop:neglogistic}
In the negative logistic model \eqref{eq:neglogistic}, the probability measure
$P_{j_0}$ for any $j_0 = 1,\dots,N$ is equal to the distribution of the random vector
\[
 \left(\frac{W_1}{W_{j_0}},\ldots,\frac{W_N}{W_{j_0}} \right)
\]
where $W_1,\ldots,W_N$ are independent, $W_j$, $j\neq j_0$, follows a 
$\mathrm{Weibull}(\theta,c_\theta)$ distribution with scale parameter 
$c_\theta=\Gamma(1+1/\theta)^{-1}$ and  $(W_{j_0}/c_\theta)^{\theta}$ follows a
$\Gamma(1+1/\theta,1)$ distribution.
\end{proposition}

\textbf{Dirichlet mixture model}

The Dirichlet mixture model was introduced by \cite{boldi-davison-2007}. In 
dimension $N$, the model corresponds to the max-stable random vector given by
\begin{equation}\label{eq:dirichlet-model}
Z=\max_{i\geq 1} \zeta_i (NY_i)
\end{equation}
where the $Y_i$'s are independent identically distributed random vectors on the
simplex 
\begin{equation*}
S_{N-1}=\left\{y\in[0,1]^n: \, \sum\nolimits_{j=1}^N y_j=1\right\}.
\end{equation*}
The distribution of each $Y_i$ is a mixture of $m$ Dirichlet models, i.e.\ its
Lebesgue density is of the form
\begin{equation} \label{eq:spec-mixt-dirichlet}
 h({y}) = \sum_{k=1}^m \pi_k \text{diri}({y} \mid \alpha_{1k}, \ldots,\alpha_{Nk}), 
                 \quad {y} = (w_1, \ldots, w_N) \in S_{N-1},
\end{equation}
where $\pi_k \geq 0$, $k = 1,\ldots,m$ such that $\sum_{k=1}^m \pi_k = 1$,
$\alpha_{ik} > 0$, $i=1,\ldots,N$, $k=1,\ldots,m$, and
\begin{equation} \label{eq:dirichlet-dens}
 \text{diri}({y} \mid \alpha_1, \ldots,\alpha_N) =  \frac{1}{B(\alpha)}\prod_{j=1}^N y_j^{\alpha_j-1},
                                  \quad B(\alpha)=\frac{\prod_{j=1}^N \Gamma(\alpha_j)}{\Gamma(\sum_{j=1}^N \alpha_j)}.
\end{equation}
Here, the parameters $\pi_k$ and $\alpha_{ik}$, $i=1,\ldots,N$, $k=1,\ldots,m$,
are such that
\[
 \mathrm{E}[Y_j]= \sum_{k=1}^m \pi_k \frac{\alpha_{jk}}{\sum_{i=1}^N \alpha_{ik}}=\frac{1}{N}, \quad j=1,\ldots,N.
\]

\begin{proposition}\label{prop:dirichlet}
In the Dirichlet model \eqref{eq:dirichlet-model}, we have for any $j_0 = 1,\dots,N$ that
$P_{j_0} = \sum_{k=1}^m \hat \pi_k P_{j_0}^{(k)}$ where 
$\hat \pi_k = \pi_k \alpha_{j_0k}/(\sum_{i=1}^N \alpha_{ik})$ and 
$P_{j_0}^{(k)}$ is equal to the distribution of the random vector
\[
 \left(\frac{G_1^{(k)}}{G_{j_0}^{(k)}},\ldots,\frac{G_N^{(k)}}{G_{j_0}^{(k)}} \right)
\]
and $G_1^{(k)},\ldots,G_N^{(k)}$ are independent random variables with Gamma
distribution
\[
 G_{j_0}^{(k)}\sim \mathrm{Gamma}(\alpha_{j_0k}+1,1)\quad \mbox{and}\quad  G_{j}\sim \mathrm{Gamma}(\alpha_{jk},1),\quad j\neq j_0.
\]
\end{proposition}

\section{Simulation on dense grids}\label{sec_simu}

In many applications, one is interested in simulating a max-stable process
$Z$ on a dense grid, e.g.\ $x = \mathcal{X}\cap (\varepsilon \mathbb{Z})^d$.
As discussed in Section \ref{sec:complexity}, on average, this requires the 
simulation of $\mathrm{E} C_2(N)=N$ random functions in Algorithm 
\ref{algo-ppp}, that is, the simulation of $N$ random vectors of size $N$.
For small $\varepsilon$, $N$ will be large and the procedure can become
very time consuming. 
Thus, one might be interested in aborting Algorithm \ref{algo-ppp} after 
$m<N$ steps, ensuring exactness of the simulation 
only at locations $x_1,\ldots,x_m$. In this case, an alternative
design of the algorithm which efficiently chooses the subset of $m$
locations might improve the probability of an exact sample
at all $N$ locations.

For comparison of two designs, we introduce the random number
\begin{equation*}
  N_0 = \min \{ m \in \{1,\ldots,N\}: Z_{m}({x}) = Z_{N}({x})\}.
\end{equation*}	
For $n\geq N_0$, the algorithm does not provide any new extremal functions,
but all the simulated functions are rejected. Hence, $N_0$ is the optimal
number of iterations before aborting the algorithm. One design
is preferable to another if its corresponding random number $N_0$ tends
to be smaller. An efficient design should thus simulate the
extremal functions at an early stage of the algorithm. Based on the intuition 
that $\phi^+_{x_{n+1}}$ is likely not to be contained in $\Phi^+_n$ if 
$Z_n(x_{n+1})$ is small, we propose the following adaptive numbering 
$x^{(1)},\dots, x^{(N)}$ of points
in Algorithm \ref{algo-ppp}:
\begin{equation}\label{eq:adapt}
 \mathrm{set}\ x^{(1)} = x_1 \ \mathrm{and}\  x^{(n+1)} = \mathrm{argmin}\left\{ Z_n(x): x \in \{x_1,\ldots,x_N\}\setminus \{x^{(1)}, \ldots, x^{(n)}\}\right\},
\end{equation}
for $n = 1,\dots, N-1$.
A simulation study indicates that this adaptive version is clearly
preferable to Algorithm \ref{algo-ppp} with a deterministic numbering
of locations. The advantage is particularly big in the case of strong
dependence which corresponds to simulation on dense grids. More details
on the simulation study and its results are provided in the supplementary
material to this paper.

\section*{Acknowledgement}

Financial support from the Swiss National Science Foundation (second author)
and the ANR project McSim (third author) is gratefully acknowledged.

\bibliography{ref}

\begin{thebibliography}{29}
\providecommand{\natexlab}[1]{#1}
\providecommand{\url}[1]{\texttt{#1}}
\expandafter\ifx\csname urlstyle\endcsname\relax
  \providecommand{\doi}[1]{doi: #1}\else
  \providecommand{\doi}{doi: \begingroup \urlstyle{rm}\Url}\fi

\bibitem[Blanchet and Davison(2011)]{bd2011}
J.~Blanchet and A.~C. Davison.
\newblock Spatial modeling of extreme snow depth.
\newblock \emph{Ann. Appl. Stat.}, 5\penalty0 (3):\penalty0 1699--1725, 2011.

\bibitem[Boldi(2009)]{bol2009}
M.-O. Boldi.
\newblock A note on the representation of parametric models for multivariate
  extremes.
\newblock \emph{Extremes}, 12:\penalty0 211--218, 2009.

\bibitem[Boldi and Davison(2007)]{boldi-davison-2007}
M.-O. Boldi and A.~C. Davison.
\newblock A mixture model for multivariate extremes.
\newblock \emph{J. R. Stat. Soc. Ser. B Stat. Methodol.}, 69\penalty0
  (2):\penalty0 217--229, 2007.

\bibitem[Buishand et~al.(2008)Buishand, de~Haan, and Zhou]{bdhz08}
T.~A. Buishand, L.~de~Haan, and C.~Zhou.
\newblock On spatial extremes: with application to a rainfall problem.
\newblock \emph{Ann. Appl. Stat.}, 2\penalty0 (2):\penalty0 624--642, 2008.

\bibitem[Cap{\'e}ra{\`a} et~al.(2000)Cap{\'e}ra{\`a}, Foug{\`e}res, and
  Genest]{cap2000}
P.~Cap{\'e}ra{\`a}, A.-L. Foug{\`e}res, and C.~Genest.
\newblock Bivariate distributions with given extreme value attractor.
\newblock \emph{J. Multivariate Anal.}, 72:\penalty0 30--49, 2000.

\bibitem[Coles(1993)]{coles93}
S.~G. Coles.
\newblock Regional modelling of extreme storms via max-stable processes.
\newblock \emph{J. R. Stat. Soc. Ser. B Stat. Methodol.}, 55\penalty0
  (4):\penalty0 797--816, 1993.

\bibitem[Davison et~al.(2012)Davison, Padoan, and Ribatet]{dav2012b}
A.~C. Davison, S.~A. Padoan, and M.~Ribatet.
\newblock Statistical modeling of spatial extremes.
\newblock \emph{Statist. Sci.}, 27:\penalty0 161--186, 2012.

\bibitem[de~Haan(1984)]{dehaan-1984}
L.~de~Haan.
\newblock A spectral representation for max-stable processes.
\newblock \emph{Ann. Probab.}, 12\penalty0 (4):\penalty0 1194--1204, 1984.

\bibitem[Dieker and Mikosch(2015)]{DM-2014}
A.~B. Dieker and T.~Mikosch.
\newblock Exact simulation of {B}rown--{R}esnick random fields at a finite
  number of locations.
\newblock \emph{Extremes}, pages 1--14, 2015.

\bibitem[Dombry and {\'E}yi-Minko(2012)]{dombry-eyiminko-2012}
C.~Dombry and F.~{\'E}yi-Minko.
\newblock Strong mixing properties of max-infinitely divisible random fields.
\newblock \emph{Stochastic Process. Appl.}, 122\penalty0 (11):\penalty0
  3790--3811, 2012.

\bibitem[Dombry and {\'E}yi-Minko(2013)]{dombry-eyiminko-2013}
C.~Dombry and F.~{\'E}yi-Minko.
\newblock Regular conditional distributions of continuous max-infinitely
  divisible random fields.
\newblock \emph{Electron. J. Probab.}, 18\penalty0 (7):\penalty0 1--21, 2013.

\bibitem[Dombry et~al.(2013)Dombry, {\'E}yi-Minko, and Ribatet]{DEMR13}
C.~Dombry, F.~{\'E}yi-Minko, and M.~Ribatet.
\newblock Conditional simulation of max-stable processes.
\newblock \emph{Biometrika}, 100\penalty0 (1):\penalty0 111--124, 2013.

\bibitem[Engelke et~al.(2011)Engelke, Kabluchko, and Schlather]{EKS-2011}
S.~Engelke, Z.~Kabluchko, and M.~Schlather.
\newblock An equivalent representation of the {B}rown--{R}esnick process.
\newblock \emph{Statist. Probab. Lett.}, 81\penalty0 (8):\penalty0 1150--1154,
  2011.

\bibitem[Engelke et~al.(2014)Engelke, Malinowski, Oesting, and
  Schlather]{eng2012}
S.~Engelke, A.~Malinowski, M.~Oesting, and M.~Schlather.
\newblock Statistical inference for max-stable processes by conditioning on
  extreme events.
\newblock \emph{Adv. Appl. Probab.}, 46:\penalty0 478--495, 2014.

\bibitem[Engelke et~al.(2015)Engelke, Malinowski, Kabluchko, and
  Schlather]{eng2014}
S.~Engelke, A.~Malinowski, Z.~Kabluchko, and M.~Schlather.
\newblock Estimation of {H}\"usler--{R}eiss distributions and
  {B}rown--{R}esnick processes.
\newblock \emph{J. R. Stat. Soc. Ser. B Stat. Methodol.}, 77:\penalty0
  239--265, 2015.

\bibitem[Ghoudi et~al.(1998)Ghoudi, Khoudraji, and Rivest]{gou1998}
K.~Ghoudi, A.~Khoudraji, and L.-P. Rivest.
\newblock Propri{\'e}t{\'e}s statistiques des copules de valeurs extr{\^e}mes
  bidimensionnelles.
\newblock \emph{Canad. J. Statist.}, 26:\penalty0 187--197, 1998.

\bibitem[Gudendorf and Segers(2010)]{gudendorf-segers-2010}
G.~Gudendorf and J.~Segers.
\newblock Extreme-value copulas.
\newblock In \emph{Copula theory and its applications}, pages 127--145.
  Springer, 2010.

\bibitem[H{\"u}sler and Reiss(1989)]{huesler-reiss-1989}
J.~H{\"u}sler and R.-D. Reiss.
\newblock Maxima of normal random vectors: between independence and complete
  dependence.
\newblock \emph{Statist. Probab. Lett.}, 7:\penalty0 283--286, 1989.

\bibitem[Kabluchko(2011)]{kab2011}
Z.~Kabluchko.
\newblock Extremes of independent {G}aussian processes.
\newblock \emph{Extremes}, 14:\penalty0 285--310, 2011.

\bibitem[Kabluchko et~al.(2009)Kabluchko, Schlather, and de~Haan]{KSH-2009}
Z.~Kabluchko, M.~Schlather, and L.~de~Haan.
\newblock Stationary max-stable fields associated to negative definite
  functions.
\newblock \emph{Ann. Probab.}, 37\penalty0 (5):\penalty0 2042--2065, 2009.

\bibitem[Oesting and Schlather(2014)]{oestingschlather14}
M.~Oesting and M.~Schlather.
\newblock Conditional sampling for max-stable processes with a mixed moving
  maxima representation.
\newblock \emph{Extremes}, 17\penalty0 (1):\penalty0 157--192, 2014.

\bibitem[Oesting et~al.(2012)Oesting, Kabluchko, and Schlather]{OKS-2012}
M.~Oesting, Z.~Kabluchko, and M.~Schlather.
\newblock Simulation of {B}rown-{R}esnick processes.
\newblock \emph{Extremes}, 15:\penalty0 89--107, 2012.

\bibitem[Oesting et~al.(2013)Oesting, Schlather, and Zhou]{OSZ-2013}
M.~Oesting, M.~Schlather, and C.~Zhou.
\newblock On the normalized spectral representation of max-stable processes on
  a compact set.
\newblock Available from \texttt{http://arxiv.org/abs/1310.1813}, 2013.

\bibitem[Opitz(2013)]{opitz-2013}
T.~Opitz.
\newblock Extremal {$t$} processes: Elliptical domain of attraction and a
  spectral representation.
\newblock \emph{J. Multivar. Anal.}, 122:\penalty0 409--413, 2013.

\bibitem[Resnick(2008)]{res2008}
S.~I. Resnick.
\newblock \emph{Extreme Values, Regular Variation and Point Processes}.
\newblock Springer, New York, 2008.

\bibitem[Schlather(2002)]{schlather-2002}
M.~Schlather.
\newblock Models for stationary max-stable random fields.
\newblock \emph{Extremes}, 5\penalty0 (1):\penalty0 33--44, 2002.

\bibitem[Smith(1990)]{smith-1990}
R.~L. Smith.
\newblock Max-stable processes and spatial extremes.
\newblock Unpublished manuscript, 1990.

\bibitem[Stephenson(2003)]{ste2003}
A.~Stephenson.
\newblock Simulating multivariate extreme value distributions of logistic type.
\newblock \emph{Extremes}, 6:\penalty0 49--59, 2003.

\bibitem[Thibaud and Opitz(2014)]{thi2014}
E.~Thibaud and T.~Opitz.
\newblock {Efficient inference and simulation for elliptical {P}areto
  processes}.
\newblock Available from \texttt{http://arxiv.org/abs/1401.0168}, 2014.

\end{thebibliography}

\pagebreak

\appendix
\section{Supplementary material}

\subsection{Proof of Proposition \ref{prop:complex}}

\begin{proof}[Proof of Proposition \ref{prop:complex}]
In order to analyze the complexity of Algorithm \ref{algo-ppp}, we consider
each step of the algorithm separately. In the $n$th step, i.e.\ for sampling
the process perfectly at site $x_n$, we simulate Poisson points $\zeta$ and 
stochastic processes $Y$, until
one of the following two conditions is satisfied:
\begin{itemize}
 \item[(a)] $\zeta < Z_{n-1}(x_n)$. This condition is checked directly after the 
       simulation of $\zeta$ and, in this case, no stochastic process
       $Y$ needs to be simulated.
 \item[(b)] $\zeta > Z_{n-1}(x_n)$ and $\zeta Y(x_i) \leq Z(x_i)$ for all 
       $1 \leq i <n-1$. In this case, $Z$ is updated and $\zeta Y$ is
       an extremal function as it contributes to $Z$ at site $x_n$ (and possibly
       also at some of the sites $x_{n+1},\ldots,x_N$).       
\end{itemize}
Thus, any stochastic process that is simulated is either rejected, i.e.\ it is
not considered as contribution to $Z$ as it does not respect all the values 
$Z(x_1), \ldots, Z(x_{n-1})$, or it leads to an extremal function. 
Denoting by $\{(\xi_i^{(n)},\psi_i^{(n)}), i \geq 1\}$ a Poisson point process
on $(0,\infty) \times \mathcal{C}$ with intensity measure $\xi^{-2} {\rm d}\xi 
\, P_{x_n}({\rm d} \psi)$, the random number $C_2(N)$ of processes simulated in
Algorithm \ref{algo-ppp} satisfies
\begin{align} \label{eq:number-decomp-2}
 C_2(N) ={}& |\Phi^{+}_{\{x_1,\ldots,x_N\}}| 
        + \sum_{n=2}^N \left|\left\{i \geq 1: \ \xi_i^{(n)} > Z(x_n), \
				  \xi_i^{(n)} > \min_{j=1}^{n-1} \frac{Z(x_j)}{\psi_i^{(n)}(x_j)}\right\}\right|.
\end{align}
In this formula, the term $|\Phi^{+}_{\{x_1,\ldots,x_N\}}|$ is the number of 
extremal functions that need to be simulated, and the term with index $n$ in 
the sum is the number of functions that are simulated but rejected since 
$\xi_i^{(n)}\psi_i^{(n)}(x_j)>Z(x_j)$ for some $j\leq n-1$. For the computation 
of the expectation of the second term, conditionally on 
$\Phi^{+}_{\{x_1,\ldots,x_{n-1}\}}$, i.e.\ for fixed $Z(x_j)$, 
$1 \leq j \leq n-1$, the two sets 
\begin{align*}
             \Phi^{(n)}_1 ={}& \{ (\xi_i^{(n)}, \psi_i^{(n)}): \, \xi_i^{(n)} \psi_i^{(n)}(x_j) > Z(x_j) 
						 \text{ for some } j=1,\ldots,n-1\}\\
 \text{and } \Phi^{(n)}_2 ={}& \{ (\xi_i^{(n)}, \psi_i^{(n)}): \, \xi_i^{(n)} \psi_i^{(n)}(x_j) \leq Z(x_j) 
             \text{ for all } j=1,\ldots,n-1\}
\end{align*}
are independent Poisson point processes with intensities
$\xi^{-2} \mathbf{1}_{\left\{\xi > \min_{j=1}^{n-1} (Z(x_j) / \psi(x_j))\right\}} {\rm d}\xi \, P_{x_n}({\rm d}\psi)$
and 
$\xi^{-2} \mathbf{1}_{\left\{\xi < \min_{j=1}^{n-1} (Z(x_j) / \psi(x_j))\right\}} {\rm d}\xi \, P_{x_n}({\rm d}\psi)$,
respectively. Conditioning further on $\Phi^{(n)}_2$, $Z(x_n)$ is also fixed
and we obtain
\begin{align*}
    & \mathrm{E} \left(\left|\left\{(\xi_i^{(n)},\psi_i^{(n)}): \, \xi_i^{(n)} > Z(x_n), \, 
                                            \xi_i^{(n)} > \min_{j=1}^{n-1} \frac{Z(x_j)}{\psi_i^{(n)}(x_j)}\right\}\right|\right)
			\displaybreak[0]\\
 ={}& \mathrm{E}\left(\mathrm{E}\left(\left|\left\{ (\xi,\psi) \in \Phi^{(n)}_1: \, \xi > Z(x_n) \right\}\right| \, \Big| \ 
                                      \Phi^{+}_{\{x_1,\ldots,x_{n-1}\}}, \, \Phi^{(n)}_2\right)\right) \displaybreak[0]\\
 ={}& \mathrm{E}\left(\int \int \xi^{-2} \mathbf{1}_{\{\xi > Z(x_n)\}} 
                                         \mathbf{1}_{\left\{\xi > \min_{j=1}^{n-1} \frac{Z(x_j)}{\psi(x_j)}\right\}} 
																				{\rm d}\xi \,  P_{x_n}({\rm d}\psi)\right) \displaybreak[0]\\
 ={}& \mathrm{E}\left(\min\left\{\frac 1 {Z(x_n)}, \max_{j=1}^{n-1} \frac {Y_n(x_j)}{Z(x_j)}\right\}\right)
\end{align*}
where $Y_n \sim P_{x_n}$ and $Z$ are independent. The relation $\min\{a,b\} = a + b - \max\{a,b\}$,
$a, b \in \mathbb{R}$, and the fact that $Y_n(x_n)=1$ almost surely yield
\begin{align*}
    & \mathrm{E} \left(\left|\left\{(\xi_i^{(n)},\psi_i^{(n)}): \, \xi_i^{(n)} > Z(x_n), \, 
                                            \xi_i^{(n)} > \min_{j=1}^{n-1} \frac{Z(x_j)}{\psi_i^{(n)}(x_j)}\right\}\right|\right)
			\displaybreak[0]\\
 ={}& \mathrm{E} \left(\frac 1 {Z(x_n)}\right) + \mathrm{E}\left(\max_{j=1}^{n-1} \frac {Y_n(x_j)}{Z(x_j)}\right)
      -  \mathrm{E}\left(\max_{j=1}^{n} \frac {Y_n(x_j)}{Z(x_j)}\right)			\displaybreak[0]\\
 ={}& 1 + \mathrm{E} | \Phi^{+}_{\{x_1,\ldots,x_{n-1}\}} | - \mathrm{E} |\Phi^{+}_{\{x_1,\ldots,x_{n}\}}|,
\end{align*}
as $\mathrm{E}|\Phi^{+}_{\{x_1,\ldots,x_n\}}| = \mathrm{E}\left(\max_{j=1}^n Y_n(x_j)/Z(x_j)\right)$
by Lemma 4.7 in \citet{OSZ-2013}. Thus, by \eqref{eq:number-decomp-2}, we obtain
\begin{align*}
 \mathrm{E} C_2(N) ={}& \mathrm{E} | \Phi^{+}_{\{x_1,\ldots,x_{N}\}} | 
                 + \sum\nolimits_{n=2}^N \left(1 + \mathrm{E} | \Phi^{+}_{\{x_1,\ldots,x_{n-1}\}} | 
								- \mathrm{E} |\Phi^{+}_{\{x_1,\ldots,x_{n}\}}|\right) \displaybreak[0]\\
                 ={}& N-1 + \mathrm{E} | \Phi^{+}_{\{x_1\}} | = N.
\end{align*}
Moreover, by \eqref{eq:normalisation}, we have that $\mathrm{E} Z(x_i)^{-1}=1$ for
 $i=1,\ldots,N$, and, thus, 
$$\mathrm{E}\left(\max_{i=1}^N Z(x_i)^{-1}\right) \geq 1,$$
with equality if only if $Z(x_1) = \ldots = Z(x_N)$ holds almost surely.
\end{proof}

\subsection{Proofs for Section \ref{sec:examples}}
\subsubsection{Moving maximum process}
\begin{proof}[Proof of Proposition \ref{prop:mmp}]
In the case of the moving maximum process $\eqref{eq:MMP}$, the measure $\nu$ associated with the representation \eqref{eq:spectralrep} is
\[
\nu(A)=\int_{\mathcal{X}} 1_{\{h(\cdot-\chi)\in A\}}\lambda(\mathrm{d}\chi),\quad A\subset \mathcal{C}\ \mbox{Borel}.
\]
We deduce from Proposition \ref{prop2}, 
\begin{eqnarray*}
P_{x_0}(A)&=&\int_{\mathcal{C}}1_{\{f/f({x_0}) \in A\}}f({x_0})\nu(\mathrm{d}f)
{}={} \int_{\mathcal{X}}1_{\{h(\cdot-\chi)/h({x_0}-\chi) \in A\}}h({x_0}-\chi)\lambda(\mathrm{d}\chi)\\
&=&\int_{\mathcal{X}}1_{\{h(\cdot+u-{x_0})/h(u) \in A\}}h(u)\lambda(\mathrm{d}u)
\end{eqnarray*}
where the last line follows from the simple change of variable ${x_0}-\chi=u$. This proves the result  since $h(u)\lambda(\mathrm{d}u)$ is a density function on $\mathcal{X}$.
\end{proof}

\subsubsection{Brown--Resnick process}

Our proof of Proposition \ref{prop:BR} relies on the following lemma on exponential changes of measures for Gaussian
processes.

\begin{lemma}\label{lem:BR}
The distribution of the random process $(W(x))_{x\in\mathcal{X}}$ under the transformed probability measure  $\widehat{\mathrm{pr}}=e^{W(x_0)-\sigma^2(x_0)/2}\mathrm{d}\mathrm{pr}$ is equal 
to the distribution of the Gaussian random process 
\[
W(x)+c(x_0,x),\quad x\in \mathcal{X},
\]
where $c(x,y)$ denotes the covariance between $W(x)$ and $W(y)$.
\end{lemma}

\begin{proof}[Proof of Lemma \ref{lem:BR}]
We need  to consider  finite dimensional distributions only and we compute for some $x_1,\ldots,x_k\in\mathcal{X}$
the Laplace transform of $(W(x_i))_{1\leq i\leq k}$ under the transformed probability measure
$\widehat{\mathrm{pr}}$. For all $\theta=(\theta_1,\ldots,\theta_k)\in\mathbb{R}^k$, we have
\begin{align}
\mathcal{L}(\theta_1,\ldots,\theta_k)={}& \widehat{\mathrm{E}}\left[e^{\sum_{i=1}^k \theta_i W(x_i)} \right]
{}={}\mathrm{E}\left[e^{W(x_0)-\sigma^2(x_0)/2}e^{\sum_{i=1}^k \theta_i W(x_i)} \right]\nonumber\\
={}& \exp\left(\frac{1}{2}\tilde\theta^\T \Sigma \tilde\theta-\frac{1}{2}\sigma^2(x_0)\right),\label{eq:Laplace}
\end{align}
with $\tilde\theta=(1,\theta)\in\mathbb{R}^{k+1}$ and $\tilde\Sigma=\left(c(x_i,x_j)\right)_{0\leq i,j\leq k}$ the covariance matrix. 
We introduce the block decomposition 
\[
\tilde\Sigma= \left(\begin{array}{cc} \sigma^2(x_0) & \Sigma_{0,k} \\ \Sigma_{k,0} & \Sigma \end{array} \right)
\]
with $\Sigma=\left(c(x_i,x_j)\right)_{1\leq i,j\leq k}$ and $\Sigma_{k,0}=\Sigma_{0,k}^\T=(c(x_0,x_i))_{1\leq i\leq k}$.
The quadratic form in Equation \eqref{eq:Laplace} can be rewritten as
\begin{align*}
\frac{1}{2}\tilde\theta^\T \tilde\Sigma \tilde\theta-\frac{1}{2}\sigma^2(x_0)
=\frac{1}{2}\left( \sigma^2(x_0) +\theta^\T\Sigma\theta+ 2\theta^\T\Sigma_{k,0} \right)-\frac{1}{2}\sigma^2(x_0)
= \theta^\T\Sigma_{k,0}+\frac{1}{2}\theta^\T\Sigma\theta.
\end{align*}
We recognize the Laplace transform of a Gaussian random vector with mean $\Sigma_{k,0}$ and covariance matrix $\Sigma$ whence the Lemma follows.
\end{proof}

\begin{proof}[Proof of Proposition \ref{prop:BR}]
Equations \eqref{eq:prop2} and \eqref{eq:BR} together with Lemma \ref{lem:BR} yield, for all Borel set $A\subset\mathcal{C}$,
\begin{eqnarray*}
P_{x_0}(A)&=& \int_{\mathcal{C}}1_{\{f/f(x) \in A\}}f(x)\nu(\mathrm{d}f)
{}={} \mathrm{E}\left[e^{W(x_0)-\frac{1}{2}\sigma^2(x_0)}1_{\{e^{W(\cdot)-\frac{1}{2}\sigma^2(\cdot)}/
               e^{W(x_0)-\frac{1}{2}\sigma^2(x_0)}\in A\}} \right]\\
&=& \widehat{\mathrm{pr}}\left[ \exp\left(W(\cdot)-W(x_0)-\frac{1}{2}(\sigma^2(\cdot)-\sigma^2(x_0))\right)\in A\right]\\
&=& \mathrm{pr}\left[ \exp\left(W(\cdot)+c(x_0,\cdot)-W(x_0)-c(x_0,x_0)-\frac{1}{2}(\sigma^2(\cdot)-\sigma^2(x_0))\right)\in A\right]\\
&=& \mathrm{pr}\left[ \exp\left(W(\cdot)-W(x_0)-\frac{1}{2}(\sigma^2(\cdot)+\sigma^2(x_0)-2c(x_0,\cdot))\right)\in A\right].
\end{eqnarray*}
Using the fact that for all $x\in\mathcal{X}$
\[
\sigma^2(x)+\sigma^2(x_0)-2c(x_0,x)=\mathrm{Var}[W(x)-W(x_0)]
\]
we deduce that $P_{x_0}$ is equal to the distribution of the log-normal process
\[
\widetilde Y(x)=\exp\left(W(x)-W(x_0)-\frac{1}{2}\mathrm{Var}[W(x)-W(x_0)] \right),\quad x\in\mathcal{X}.
\]
This proves Proposition \ref{prop:BR}.
\end{proof}

\subsubsection{Extremal-$t$ process}

In the sequel, we write shortly $z^\alpha=\max(0,z)^\alpha$ for all real numbers $z$.

\begin{lemma}\label{lem:extremal-t}
The distribution of the random process $\left(W(x)/W(x_0)\right)_{x\in\mathcal{X}}$ 
under the transformed probability measure  $\widehat{\mathrm{pr}}=c_\alpha W(x_0)^\alpha\mathrm{d}\mathrm{pr}$ is equal to the distribution of a Student 
process with $\alpha+1$ degrees of freedom, location parameter $\mu_k$ and scale matrix $\widehat\Sigma_k$ given by
\[
\mu_k=\Sigma_{k,0}\quad\mbox{and}\quad\widehat\Sigma_k=\frac{\Sigma_k-\Sigma_{k,0}\Sigma_{0,k}}{\alpha+1}, 
\]
where $\Sigma_k=\left(c(x_i,x_j)\right)_{1\leq i,j\leq k}$ and $\Sigma_{k,0}=\Sigma_{0,k}^\T=(c(x_0,x_i))_{1\leq i\leq k}$.
\end{lemma}

\begin{proof}[Proof of Lemma \ref{lem:extremal-t}]
We consider finite dimensional distributions only. Let $k\geq 1$ and $x_1,\ldots,x_k\in\mathcal{X}$. 
We first assume that the covariance matrix $\widetilde \Sigma=\left(c(x_i,x_j)\right)_{0\leq i,j\leq k}$ is non singular so that $(W(x_i))_{0\leq i\leq k}$
has density
\[
\tilde g({y})=(2\pi)^{-(k+1)/2} \mathrm{det}(\widetilde\Sigma)^{-1/2}\exp\left(-\frac{1}{2}{y}^\T\widetilde\Sigma^{-1}{y}\right)\quad \mbox{with }{y}=(y_i)_{0\leq i\leq k}.
\]
Setting ${z}=(y_i/y_0)_{1\leq i\leq k}$, we have for all Borel sets $A_1,\ldots, A_k\subset\mathbb{R}$
\begin{eqnarray*}
\widehat{\mathrm{pr}}\left[\frac{W(x_i)}{W(x_0)}\in A_i,\ i=1,\ldots,k\right]
&=&\int_{\mathbb{R}^{k+1}} 1_{\{y_i/y_0\in A_i,\ i=1,\ldots,k\}}c_\alpha y_0^\alpha \tilde g({\bf y})\,\mathrm{d}{y}\\
&=&\int_{\mathbb{R}^{k}}1_{\{z_i\in A_i,\ i=1,\ldots,k\}}\left( \int_0^\infty c_\alpha y_0^\alpha  \tilde g(y_0,y_0{z})\,y_0^k \mathrm{d}y_0\right)\mathrm{d}{z}
\end{eqnarray*}
We deduce that under $\widehat{\mathrm{pr}}$, the random vector $(W(x_i)/W(x_0))_{1\leq i\leq k}$ has density
\begin{eqnarray*}
g({z})&=&\int_0^\infty c_\alpha y_0^{k+\alpha}  \tilde{g}(y_0,y_0{z})\,\mathrm{d}y_0\\
&=& c_\alpha(2\pi)^{-(k+1)/2} \mathrm{det}(\widetilde\Sigma)^{-1/2}\int_0^\infty  y_0^{k+\alpha} \exp\left(-\frac{\tilde{{z}}^\T\widetilde\Sigma^{-1}\tilde{{z}}}{2}y_0^2\right) \mathrm{d}y_0
\end{eqnarray*}
with $\tilde{{z}}=(1,{z})$. Using the change of variable $u=\frac{\tilde{{z}}^\T\widetilde\Sigma^{-1}\tilde{{z}}}{2}y_0^2$, we get
\begin{eqnarray*}
\int_0^\infty  y_0^{k+\alpha} \exp\left(-\frac{\tilde{{z}}^\T\widetilde\Sigma^{-1}\tilde{{z}}}{2}y_0^2\right) \mathrm{d}y_0
&=& \frac{1}{2} \left(\frac{\tilde{{z}}^\T\widetilde\Sigma^{-1}\tilde{{z}}}{2}\right)^{-\frac{\alpha+k+1}{2}}\int_0^\infty  u^{(k+\alpha-1)/2} \exp\left(-u\right) \mathrm{d}u\\
&=& \frac{1}{2} \left(\frac{\tilde{{z}}^\T\widetilde\Sigma^{-1}\tilde{{z}}}{2}\right)^{-\frac{\alpha+k+1}{2}}\Gamma\left(\frac{k+\alpha+1}{2}\right)
\end{eqnarray*}
and we obtain after simplification
\[
g({z})=\pi^{-k/2} \frac{\Gamma\left(\frac{k+\alpha+1}{2}\right)}{\Gamma\left(\frac{\alpha+1}{2}\right)}\mathrm{det}(\widetilde\Sigma)^{-1/2}\left(\tilde{{z}}^\T\widetilde\Sigma^{-1}\tilde{{z}}\right)^{-\frac{\alpha+k+1}{2}}.
\]
Introducing the block decomposition 
$\widetilde\Sigma= \left(\begin{array}{cc} 1 & \Sigma_{0,k} \\ \Sigma_{k,0} & \Sigma_k \end{array} \right)$,
the inverse matrix is
\[
\widetilde\Sigma^{-1}= \left(\begin{array}{cc} 
1+\Sigma_{0,k}(\Sigma_k-\Sigma_{k,0}\Sigma_{0,k})^{-1}\Sigma_{k,0} & -\Sigma_{0,k}(\Sigma_k-\Sigma_{k,0}\Sigma_{0,k})^{-1} \\ 
-(\Sigma_k-\Sigma_{k,0}\Sigma_{0,k})^{-1}\Sigma_{k,0} & (\Sigma_k-\Sigma_{k,0}\Sigma_{0,k})^{-1} 
\end{array} \right).
\]
By the definition of $\mu_k$ and $\widehat{\Sigma}_k$, we have
\[
\widetilde\Sigma^{-1}= \frac{1}{1+\alpha}\left(\begin{array}{cc} 
1+\alpha+\mu_k^\T\widehat\Sigma_k^{-1}\mu_k & -\mu_k^\T\widehat\Sigma_k^{-1} \\ 
-\widehat\Sigma_k^{-1}\mu_k & \widehat\Sigma_k^{-1} 
\end{array} \right)
\]
and
\[
\tilde{{z}}^\T\widetilde\Sigma^{-1}\tilde{{z}}= (1,{z})^\T\widetilde\Sigma^{-1}(1,{z})
=\left(1+\frac{({z}-\mu_k)^\T\widehat\Sigma^{-1}_k({z}-\mu_k)}{\alpha+1}\right)
\]
Finally, we obtain after simplification
\[
g({z})=\pi^{-k/2} (\alpha+1)^{-k/2} \frac{\Gamma\left(\frac{k+\alpha+1}{2}\right)}{\Gamma\left(\frac{\alpha+1}{2}\right)}
\mathrm{det}(\widehat\Sigma_k)^{-1/2}\left(1+\frac{({z}-\mu_k)^\T\widehat\Sigma^{-1}_k({z}-\mu_k)}{\alpha+1}\right)^{-\frac{\alpha+k+1}{2}}.
\]
We recognize the $k$-variate Student density with $\alpha+1$ degrees of freedom, location 
parameter $\mu$ and scale matrix $\hat\Sigma_k$. 
\end{proof}

\begin{proof}[Proof of Proposition \ref{prop:extremal-t}]
Consider the  set
\[
A=\{f\in\mathcal{C}_0;f(x_1)\in A_1,\cdots,f(x_k)\in A_k\}.
\]
Equations \eqref{eq:prop2} and \eqref{eq:extremal-t} together with Lemma \ref{lem:extremal-t} yield,
\begin{eqnarray*}
P_{x_0}(A)&=& \int_{\mathcal{C}}1_{\{f/f(x) \in A\}}f(x)\nu(\mathrm{d}f)
{}={} \mathrm{E}\left[c_\alpha W(x_0)^\alpha 1_{\{W(x_i)^\alpha/W(x_0)^\alpha \in A_i,\ i=1,\ldots,k \}} \right]\\
&=& \widehat{\mathrm{pr}}\left[ W(x_i)^\alpha/W(x_0)^\alpha \in A_i,\ i=1,\ldots,k \right]\\
&=& \mathrm{pr}\left[ T_i^\alpha\in A_i,\ i=1,\ldots,k \right]
\end{eqnarray*}
where $T=(T_1,\ldots,T_k)$ has a multivariate Student distribution with $\alpha+1$ degrees of freedom, location
parameter $\mu_k$ and dispersion matrix $\widehat\Sigma_k$. This proves the result.
\end{proof}

\subsubsection{Multivariate extreme value distributions}

{\bf Logistic model}

\begin{proof}[Proof of Proposition \ref{prop:logistic}]
It is easily shown that the logistic model admits the representation
\[
 Z=\max_{i\geq 1} \zeta_i F_i
\]
where the $F_i$'s are independent random vectors with independent $\mathrm{Frechet}(\beta,c_\beta)$-distributed components.
To check this, we compute
\begin{align*}
   & \mathrm{E}\bigg[\max_{j=1}^N \frac{F_j}{z_j} \bigg]
 {}={} \int_0^\infty \mathrm{pr}\left[\max_{j=1}^N \frac{F_j}{z_j}>u\right]\mathrm{d}u
 {}={} \int_0^\infty \left(1-\prod\nolimits_{i=1}^N \mathrm{pr}[F_j<z_ju]\right)\mathrm{d}u\\
 ={}& \int_0^\infty \left(1-\prod\nolimits_{i=1}^N  e^{-(z_j u/c_\beta)^{-\beta}}\right)\mathrm{d}u=\int_0^\infty 
   \left(1- e^{-u^{-\beta}\sum\nolimits_{j=1}^N (z_j/c_\beta)^{-\beta}}\right)\mathrm{d}u
 {}={} \left(\sum\nolimits_{j=1}^N z_j^{-\beta}\right)^{1/\beta}.
\end{align*}
For the computation of the last integral, we recognize the expectation of a Fr\'echet distribution.
Next we use the fact that $P_{j_0}$ is the distribution of $F/F_{j_0}$ under the transformed density
\[
 y_{j_0}\prod_{k=1}^N \frac{\beta}{c_\beta}\left(\frac{y_k}{c_\beta}\right)^{-1-\beta}e^{-(y_k/c_\beta)^{-\beta}}.
\]
We recognize a product measure where the $j$th margin, $j\neq j_0$, has a $\mathrm{Frechet}(\beta,c_\beta)$
distribution. The $j_0$th marginal has density
\[
 y_{j_0} \frac{\beta}{c_\beta}\left(\frac{y_{j_0}}{c_\beta}\right)^{-1-\beta}e^{-(y_{j_0}/c_\beta)^{-\beta}}
\]
and a simple change of variable reveals that this is the density of $c_\beta Z^{-1/\beta}$
with $Z\sim\mathrm{Gamma}(1-1/\beta,1)$.
\end{proof}

{\bf Negative logistic model}

\begin{proof}[Proof of Proposition \ref{prop:neglogistic}]
Similarly to the logistic model, we have the spectral representation 
\[
 Z=\max_{i\geq 1} \zeta_i W_i
\]
where the $W_i$'s are independent random vectors with independent $\mathrm{Weibull}(\theta,c_\theta)$-distributed components with scale parameter $c_\theta=\frac{1}{\Gamma(1+1/\theta)}$. To check this, we compute
\begin{align*}
    \mathrm{E}\bigg[\max_{j=1}^N & \frac{W_j}{z_j} \bigg]
 {}={} \int_0^\infty \mathrm{pr}\left[\max_{j=1}^N \frac{W_j}{z_j}>u\right]\mathrm{d}u
 {}={} \int_0^\infty \left(1-\prod_{j=1}^N \mathrm{pr}[W_j<z_ju]\right)\mathrm{d}u\\
 ={}& \int_0^\infty \left(1-\prod_{j=1}^N \left(1- e^{-(z_ju/c_\theta)^\theta}\right)\right)\mathrm{d}u
 {}={} -\sum_{J}(-1)^{|J|}\int_0^\infty  e^{-u^\theta\sum_{j\in J}(z_j/c_\theta)^\theta}\mathrm{d}u\\
 ={}& -\sum_{J}(-1)^{|J|} \left(\sum\nolimits_{j\in J}(z_j/c_\theta)^\theta\right)^{-1/\theta}\Gamma(1+1/\theta)
 {}={} -\sum_{J}(-1)^{|J|} \left(\sum\nolimits_{j\in J}z_j^\theta\right)^{-1/\theta}.
\end{align*}
For the computation of the last integral, we recognize the expectation of a Weibull distribution.
As for the logistic model, $P_{j_0}$ is the distribution of $W/W_{j_0}$ under the transformed density
\[
 y_{j_0}\prod_{k=1}^N \frac{\theta}{c_\theta}\left(\frac{y_k}{c_\theta}\right)^{\theta-1}e^{-(y_k/c_\theta)^{\theta}}.
\]
We recognize a product measure where the $j$th margin, $j\neq j_0$ has a $\mathrm{Weibull}(\theta,c_\theta)$
distribution. The $j_0$th marginal has density
\[
 y_{j_0} \frac{\theta}{c_\theta}\left(\frac{y_{j_0}}{c_\theta}\right)^{\theta-1}e^{-(y_{j_0}/c_\theta)^{\theta}}
\]
and a simple change of variable reveals that this is the density of $c_\theta Z^{1/\theta}$
with $Z\sim\mathrm{Gamma}(1+1/\theta,1)$.
\end{proof}

{\bf Dirichlet mixture model}

\begin{proof}[Proof of Proposition \ref{prop:dirichlet}]
By definition, $P_{j_0}$ has the form
\begin{align*}
 P_{j_0}(A) ={}& N \sum_{k=1}^m \pi_k \int_{S_{N-1}} y_{j_0} \mathbf{1}_{\{{y}/y_{j_0} \in A\}} 
                                              \text{diri}({y} \mid \alpha_{1k}, \ldots,\alpha_{Nk}) \, {\rm d} {y}\\
            ={}& N\sum_{k=1}^m \hat \pi_k \frac{\int_{S_{N-1}} y_{j_0} \mathbf{1}_{\{{y}/y_{j_0} \in A\}} 
                                              \text{diri}({y} \mid \alpha_{1k}, \ldots,\alpha_{Nk}) \, {\rm d} {y}}
                                             {\int_{S_{N-1}} y_{j_0} \text{diri}({y} \mid \alpha_{1k}, \ldots,\alpha_{Nk}) \, {\rm d} {y}},
          \quad A \subset (0,\infty)^N.
\end{align*}
Thus, $P_{j_0}$ is given as the mixture $P_{j_0} = \sum_{k=1}^m \hat \pi_k P_{j_0}^{(k)}$,
where for each $k = 1, \ldots, m$, the probability measure 
$P_{j_0}^{(k)}$ is equal to the distribution of the random vector 
$\tilde Y^{(k)} / \tilde Y_{j_0}^{(k)}$, and $\tilde Y^{(k)}$ has a 
transformed density proportional to $y_{j_0}\prod_{j=1}^N y_j^{\alpha_j-1}$.
We recognize the Dirichlet distribution with parameters $\tilde \alpha_{1k}, \ldots, \tilde \alpha_{Nk}$
given by
\[
 \tilde\alpha_{j_0k}=\alpha_{j_0k}+1\quad\mbox{and}\quad  \tilde\alpha_{jk}=\alpha_{jk}\quad j\neq j_0.
\]
It is well known that Dirichlet distributions can be expressed in terms of Gamma distributions. More
precisely, we have the stochastic representation 
$$\tilde Y^{(k)} = \left( G_1^{(k)} \Big/ \sum_{j=1}^N G_j^{(k)},\ldots, G_N^{(k)} \Big/ \sum_{j=1}^N G_j^{(k)}\right),$$
where $G_j^{(k)}$ are independent $\mathrm{Gamma}(\tilde\alpha_{jk},1)$ random variables.
The result follows since $P_{j_0}^{(k)}$ is the distribution of $\tilde Y^{(k)} / \tilde Y_{j_0}^{(k)}$.
\end{proof}

\subsection{Simulation Study}

We perform a simulation study to compare the adaptive version of
Algorithm \ref{algo-ppp} introduced in \eqref{eq:adapt} to a version, where
the numbering of locations is deterministic. The simulation study is based on
$5000$ simulations of a Brown--Resnick process associated to a semi-variogram 
of the type $\gamma(h) = c \|h\|^\alpha$ on the two-dimensional grid 
$\{0.05,0.15,\ldots,0.95\} \times \{0.05,0.15,\ldots,0.95\}$. We run Algorithm 
\ref{algo-ppp} with the deterministic design (the grid points are ordered by 
their coordinates in the lexicographical sense) and with the 
adaptive design \eqref{eq:adapt}. The simulation is repeated for
different values of the parameter vector $(c,\alpha)$ representing
strong dependence $((c,\alpha) = (1,1.5))$, moderate dependence 
$((c,\alpha)=(2.5,1))$ and weak dependence $((c,\alpha)=(5,0.5))$. The 
histograms of $N_0$ are shown in Figure \ref{fig:hist}. For each of the 
three parameter vectors, the number $N_0$ for the adaptive design is 
stochastically smaller than the corresponding number for the deterministic
design.

\begin{figure} \label{fig:hist}
 \centering \includegraphics[width=0.99\textwidth]{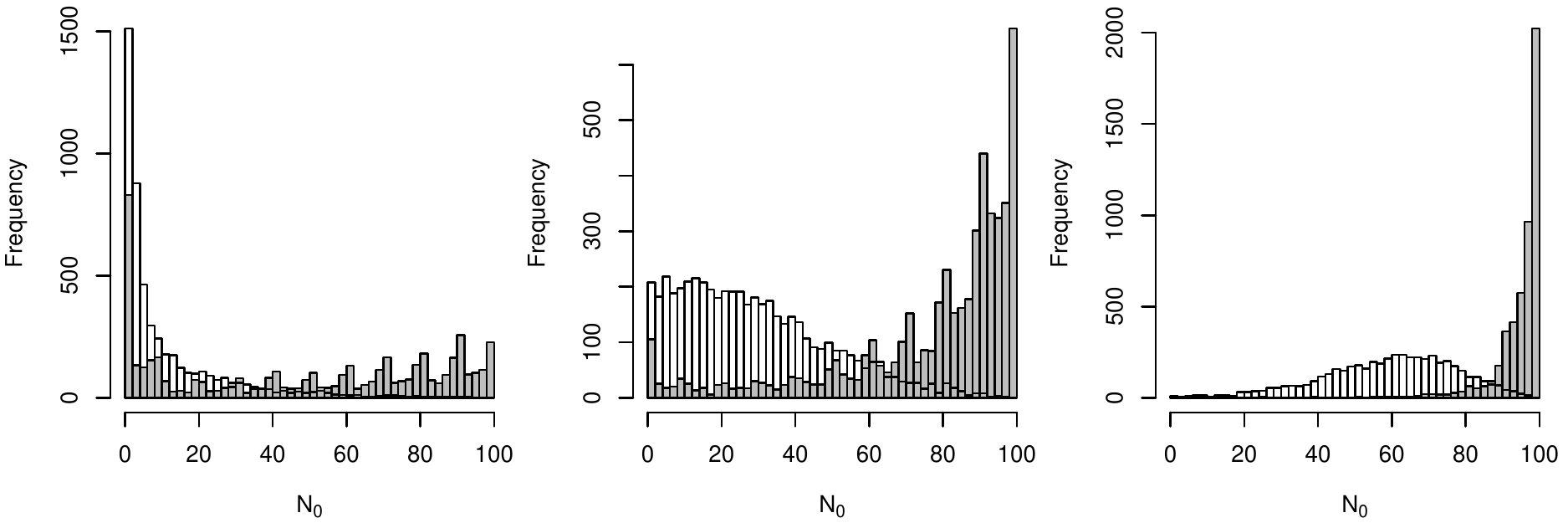}
 \caption{Histogram of $N_0$ based on 5000 realizations of a Brown--Resnick 
          process associated to the semi-variogram $\gamma(h) = c \|h\|^\alpha$
					with $c=1$ and $\alpha=1.5$ (left), $c=2.5$ and $\alpha=1$ (middle)
					and $c=5$ and $\alpha=0.5$ simulated via Algorithm \ref{algo-ppp} 
					with the deterministic design (grey) and the adaptive design 
					\eqref{eq:adapt} (white), respectively.}
\end{figure}

\end{document}